\documentclass[10pt]{article}

\usepackage[english]{babel}
\usepackage{graphicx}
\usepackage{framed,calc}
\usepackage{amsmath,amsthm,amssymb,amsfonts,algorithm,stmaryrd}
\usepackage[italicdiff]{physics}
\usepackage[T1]{fontenc}
\usepackage{wrapfig,etoolbox}
\usepackage{lmodern,mathrsfs}
\usepackage[inline,shortlabels]{enumitem}
\setlist{topsep=2pt,itemsep=2pt,parsep=0.1pt,partopsep=0.1pt}
\usepackage[dvipsnames]{xcolor}
\usepackage[utf8]{inputenc}
\usepackage{csquotes}
\usepackage[most]{tcolorbox}
\usepackage{float,sidenotes}
\usepackage{tikz,tikz-3dplot,tikz-cd,tkz-tab,tkz-euclide,pgf,pgfplots,adjustbox}
\usetikzlibrary{math,patterns,positioning,arrows,arrows.meta,graphs,graphs.standard,decorations.pathmorphing,tikzmark,calc}

\pgfplotsset{compat=newest}
\usepackage{multicol,varwidth}
\usepackage{outlines}
\usepackage{cite}

\usepackage{hyperref}
\hypersetup{
    colorlinks=true,
    linkcolor=Blue,
    filecolor=Magenta,      
    urlcolor=Cyan,
    citecolor=Green,
    pdftitle={A PAS for a special case of 2D knapsack},
    pdfpagemode=FullScreen,
    }
\usepackage[nameinlink]{cleveref} %nameinlink ensures that the entire element is clickable in the pdf, not just the number

\usepackage[a4paper, left=2.5cm, bottom=2.5cm, right=2.5cm, top=1cm]{geometry}
\usepackage{algorithmicx}
\usepackage[noend]{algpseudocode}
\colorlet{kw}{RoyalBlue}
\definecolor{com}{rgb}{0,0.6,0.3}
\algrenewcommand\algorithmicfunction{\textcolor{kw}{\textbf{function}}}
\algrenewcommand\algorithmicwhile{\textcolor{kw}{\textbf{while}}}
\algrenewcommand\algorithmicfor{\textcolor{kw}{\textbf{for}}}
\algrenewcommand\algorithmicif{\textcolor{kw}{\textbf{if}}}
\algrenewcommand\algorithmicelse{\textcolor{kw}{\textbf{else}}}
\algrenewcommand\algorithmicreturn{\textcolor{kw}{\textbf{return}}}

\algrenewcommand\algorithmicthen{\!\textcolor{kw}{\textbf{:}}}
\algrenewcommand\algorithmicdo{}
\algnewcommand\Break{\textcolor{kw}{\textbf{break}}}%
\algnewcommand\Continue{\textcolor{kw}{\textbf{continue}}}%
\algdef{SE}[WHILE]{While}{EndWhile}[1]{\algorithmicwhile\textcolor{kw}{\textbf{:}}\ #1}{}%
\algdef{SE}[FOR]{For}{EndFor}[1]{\algorithmicfor\textcolor{kw}{\textbf{:}}\ #1}{}%
\algdef{C}[IF]{If}{ElsIf}[1]{\algorithmicelse\algorithmicif\ #1 \textcolor{kw}{\textbf{:}}\ }
\algdef{C}[IF]{If}{Else}[1]{\algorithmicelse\textcolor{kw}{\textbf{:}}\ #1}
\algnotext{EndFor}
\algnotext{EndWhile}
\algrenewcommand{\algorithmiccomment}[1]{{\color{com}\%#1}}
\errorcontextlines\maxdimen
\newcommand{\ALGtikzmarkcolor}{RoyalBlue}% customise this, if you want
\newcommand{\ALGtikzmarkextraindent}{4pt}% customise this, if you want
\newcommand{\ALGtikzmarkverticaloffsetstart}{-.5ex}% customise this, if you want
\newcommand{\ALGtikzmarkverticaloffsetend}{-.5ex}% customise this, if you want
\makeatletter
\newcounter{ALG@tikzmark@tempcnta}

\newcommand\ALG@tikzmark@start{%
    \global\let\ALG@tikzmark@last\ALG@tikzmark@starttext%
    \expandafter\edef\csname ALG@tikzmark@\theALG@nested\endcsname{\theALG@tikzmark@tempcnta}%
    \tikzmark{ALG@tikzmark@start@\csname ALG@tikzmark@\theALG@nested\endcsname}%
    \addtocounter{ALG@tikzmark@tempcnta}{1}%
}

\def\ALG@tikzmark@starttext{start}
\newcommand\ALG@tikzmark@end{%
    \ifx\ALG@tikzmark@last\ALG@tikzmark@starttext
    \else
        \tikzmark{ALG@tikzmark@end@\csname ALG@tikzmark@\theALG@nested\endcsname}%
        \tikz[overlay,remember picture] \draw[\ALGtikzmarkcolor,thick] let \p{S}=($(pic cs:ALG@tikzmark@start@\csname ALG@tikzmark@\theALG@nested\endcsname)+(\ALGtikzmarkextraindent,\ALGtikzmarkverticaloffsetstart)$), \p{E}=($(pic cs:ALG@tikzmark@end@\csname ALG@tikzmark@\theALG@nested\endcsname)+(\ALGtikzmarkextraindent,\ALGtikzmarkverticaloffsetend)$) in (\x{S},\y{S})--(\x{S},\y{E});%
    \fi
    \gdef\ALG@tikzmark@last{end}%
}

\apptocmd{\ALG@beginblock}{\ALG@tikzmark@start}{}{\errmessage{failed to patch}}
\pretocmd{\ALG@endblock}{\ALG@tikzmark@end}{}{\errmessage{failed to patch}}
\makeatother
\newcommand{\hide}[1]{} %To hide large blocks of code without using % symbols

\newcommand{\eps}{\varepsilon}

\newcommand{\ceil}[1]{\ensuremath{\left\lceil#1\right\rceil}}

\newcommand{\C}{\mathcal{C}}

\newcommand{\F}{\mathcal{F}}

\newcommand{\Q}{\mathcal{Q}}
\newcommand{\R}{\mathcal{R}}
\renewcommand{\S}{\mathcal{S}}

\renewcommand{\O}{\mathcal{O}}

\newcommand{\roundell}[1]{\text{round}_\ell(#1)}
\newcommand{\col}{\text{col}}
\newcommand{\container}[1]{\text{container}(#1)}
\newcommand{\reduce}{\text{reduce}}
\newcommand{\DP}[1]{\text{dp}[#1]}

\renewcommand{\P}{\mathcal{P}}

\definecolor{BgYellow}{HTML}{FFF59C}
\definecolor{FrameYellow}{HTML}{F7A600}
\definecolor{BgPink}{HTML}{EF6FA7}
\definecolor{FramePink}{HTML}{E5446E}
\definecolor{BgGreen}{HTML}{C7D92D}
\definecolor{FrameGreen}{HTML}{89B23B}
\definecolor{BgBlue}{HTML}{45BEE9}
\definecolor{FrameBlue}{HTML}{31A8C9}
\definecolor{BgWhite}{HTML}{D8D8D8}
\definecolor{FrameWhite}{HTML}{7F7F7F}
\definecolor{BgBrown}{HTML}{8E7A45}
\definecolor{FrameBrown}{HTML}{6B5B32}

\definecolor{Bordeaux}{RGB}{116, 2, 51}

\usepackage{contour}

\newtheoremstyle{mystyle}{}{}{}{}{\sffamily\bfseries}{:}{ }{}
\newtheoremstyle{cstyle}{}{}{}{}{\sffamily\bfseries}{:}{ }{}
\makeatletter
\renewenvironment{proof}[1][\proofname] {\par\pushQED{\qed}{\normalfont\sffamily\bfseries\topsep6\p@\@plus6\p@\relax #1\@addpunct{:} }}{\popQED\endtrivlist\@endpefalse}
\makeatother
\newcommand{\coolqed}[1]{\ensuremath{\square}}
\theoremstyle{mystyle}{\newtheorem{definition}{Definition}}
\theoremstyle{mystyle}{\newtheorem{proposition}[definition]{Proposition}}
\theoremstyle{mystyle}{\newtheorem{theorem}[definition]{Theorem}}
\theoremstyle{mystyle}{\newtheorem{lemma}[definition]{Lemma}}
\theoremstyle{mystyle}{}
\theoremstyle{mystyle}{}
\theoremstyle{mystyle}{}
\theoremstyle{mystyle}{}
\theoremstyle{mystyle}{}
\theoremstyle{definition}{}
\theoremstyle{cstyle}{}

\newtheoremstyle{warn}{}{}{}{}{\normalfont}{}{ }{}
\theoremstyle{warn}

\newcommand{\warningsign}[1]{\tikz[scale=#1,every node/.style={transform shape}]{\draw[-,line width={#1*0.8mm},red,fill=yellow,rounded corners={#1*2.5mm}] (0,0)--(1,{-sqrt(3)})--(-1,{-sqrt(3)})--cycle;
\node at (0,-1) {\fontsize{48}{60}\selectfont\bfseries!};}}

\tcolorboxenvironment{definition}{boxrule=0pt,boxsep=0pt,colback={green!15},left=8pt,right=8pt,enhanced jigsaw, borderline west={2pt}{0pt}{Green},sharp corners,before skip=10pt,after skip=10pt,breakable,after={\par\smallskip\noindent}}
\tcolorboxenvironment{proposition}{boxrule=0pt,boxsep=0pt,colback={Dandelion!20},left=8pt,right=8pt,enhanced jigsaw, borderline west={2pt}{0pt}{Dandelion},sharp corners,before skip=10pt,after skip=10pt,breakable,after={\par\smallskip\noindent}}
\tcolorboxenvironment{theorem}{boxrule=0pt,boxsep=0pt,colback={OrangeRed!15},left=8pt,right=8pt,enhanced jigsaw, borderline west={2pt}{0pt}{OrangeRed},sharp corners,before skip=10pt,after skip=10pt,breakable,after={\par\smallskip\noindent}}
\tcolorboxenvironment{lemma}{boxrule=0pt,boxsep=0pt,colback={Orange!20},left=8pt,right=8pt,enhanced jigsaw, borderline west={2pt}{0pt}{Orange},sharp corners,before skip=10pt,after skip=10pt,breakable,after={\par\smallskip\noindent}}
\tcolorboxenvironment{corollary}{boxrule=0pt,boxsep=0pt,colback={violet!10},left=8pt,right=8pt,enhanced jigsaw, borderline west={2pt}{0pt}{violet},sharp corners,before skip=10pt,after skip=10pt,breakable,after={\par\smallskip\noindent}}
\tcolorboxenvironment{proof}{boxrule=0pt,boxsep=0pt,blanker,borderline west={2pt}{0pt}{CadetBlue!80!white},left=8pt,right=8pt,sharp corners,before skip=10pt,after skip=10pt,breakable,after={\par\smallskip\noindent}}
\tcolorboxenvironment{remark}{boxrule=0pt,boxsep=0pt,blanker,borderline west={2pt}{0pt}{Green},left=8pt,right=8pt,before skip=10pt,after skip=10pt,breakable,after={\par\smallskip\noindent}}
\tcolorboxenvironment{remarks}{boxrule=0pt,boxsep=0pt,blanker,borderline west={2pt}{0pt}{Green},left=8pt,right=8pt,before skip=10pt,after skip=10pt,breakable,after={\par\smallskip\noindent}}
\tcolorboxenvironment{example}{boxrule=0pt,boxsep=0pt,blanker,borderline west={2pt}{0pt}{RoyalBlue},left=8pt,right=8pt,sharp corners,before skip=10pt,after skip=10pt,breakable,after={\par\smallskip\noindent}}
\tcolorboxenvironment{examples}{boxrule=0pt,boxsep=0pt,blanker,borderline west={2pt}{0pt}{RoyalBlue},left=8pt,right=8pt,sharp corners,before skip=10pt,after skip=10pt,breakable,after={\par\smallskip\noindent}}

\tcolorboxenvironment{customthm}{boxrule=0pt,boxsep=0pt,colback={OrangeRed!15},left=8pt,right=8pt,enhanced jigsaw, borderline west={2pt}{0pt}{OrangeRed},sharp corners,before skip=10pt,after skip=10pt,breakable,after={\par\smallskip\noindent}}
\newenvironment{talign*}{\csname align*\endcsname}{\endalign}

\usepackage[explicit,noindentafter]{titlesec}
\newlength{\myl}

\titleformat{\section}[block]{}{}{4pt}{
$ $\\
\settowidth{\myl}{\fontsize{24}{30}0\sffamily\bfseries\thesection0}
\setlength{\myl}{\dimexpr \myl-3pt}
\begin{tcolorbox}
    [enhanced jigsaw,
     borderline south={2pt}{0pt}{black!70},
     sharp corners,
     boxsep=0pt,
     colback=white,
     boxrule=0pt,
     leftrule=\myl,
     colframe=black!70,
     overlay unbroken and first ={%
     \node[minimum width=1cm,
        anchor=west,
        font=\Large\sffamily\bfseries,
        white]
     at (frame.west) {\fontsize{24}{30}\sffamily\bfseries\thesection};
     }
    ]
    \fontsize{24}{30}\sffamily\bfseries#1
\end{tcolorbox}
}

\titleformat{name=\section,numberless}[block]{}{}{4pt}{
$ $\\
\begin{tcolorbox}
    [enhanced jigsaw,boxrule=0pt,boxsep=0pt,borderline west={5pt}{0pt}{black!70},borderline south={2pt}{0pt}{black!70},left=8pt,right=8pt,sharp corners,before skip=10pt,after skip=10pt,colback=white,
    ]
    \fontsize{24}{30}\sffamily\bfseries#1
\end{tcolorbox}
}

\titleformat{\subsection}[block]{}{}{4pt}{
$ $\\
\settowidth{\myl}{0\fontsize{16}{18}{\sffamily\bfseries\thesubsection}0}
\setlength{\myl}{\dimexpr \myl-3pt}
\begin{tcolorbox}
    [enhanced jigsaw,
     borderline south={2pt}{0pt}{black!70},
     sharp corners,
     boxsep=0pt,
     colback=white,
     boxrule=0pt,
     leftrule=\myl,
     colframe=black!70,
     overlay unbroken and first ={%
     \node[minimum width=1cm,
        anchor=west,
        font=\Large\sffamily\bfseries,
        white]
     at (frame.west) {\fontsize{16}{18}\sffamily\bfseries\thesubsection};
     }
    ]
    \fontsize{16}{18}\sffamily\bfseries#1
\end{tcolorbox}
}

\titleformat{name=\subsection,numberless}[block]{}{}{4pt}{
$ $\\
\begin{tcolorbox}
    [enhanced jigsaw,boxrule=0pt,boxsep=0pt,borderline west={5pt}{0pt}{black!70},borderline south={2pt}{0pt}{black!70},left=8pt,right=8pt,sharp corners,before skip=10pt,after skip=10pt,colback=white,
    ]
    \fontsize{16}{18}\sffamily\bfseries#1
\end{tcolorbox}
}

\titleformat{\subsubsection}[block]{}{}{4pt}{
$ $\\
\settowidth{\myl}{\fontsize{10}{12}0\sffamily\bfseries\thesubsubsection0}
\setlength{\myl}{\dimexpr \myl-3pt}
\begin{tcolorbox}
    [enhanced jigsaw,
     borderline south={2pt}{0pt}{black!70},
     sharp corners,
     boxsep=0pt,
     colback=white,
     boxrule=0pt,
     leftrule=\myl,
     colframe=black!70,
     overlay unbroken and first ={%
     \node[minimum width=1cm,
        anchor=west,
        font=\Large\sffamily\bfseries,
        white]
     at (frame.west) {\fontsize{10}{12}\sffamily\bfseries\thesubsubsection};
     }
    ]
    \fontsize{10}{12}\sffamily\bfseries#1
\end{tcolorbox}
}

\titleformat{name=\subsubsection,numberless}[block]{}{}{4pt}{
$ $\\
\begin{tcolorbox}
    [enhanced jigsaw,boxrule=0pt,boxsep=0pt,borderline west={5pt}{0pt}{black!70},borderline south={2pt}{0pt}{black!70},left=8pt,right=8pt,sharp corners,before skip=10pt,after skip=10pt,colback=white,
    ]
    \fontsize{10}{12}\sffamily\bfseries#1
\end{tcolorbox}
}

\titlespacing*{\section}{0pt}{5pt}{5pt}
\titlespacing*{\subsection}{0pt}{5pt}{5pt}
\titlespacing*{\subsubsection}{0pt}{5pt}{5pt}

\title{A parameterized approximation scheme for the 2D-Knapsack problem with wide items}
\title{\huge\sffamily\bfseries A parameterized approximation scheme for the 2D-Knapsack problem with wide items}
\author{
  {\sffamily \bfseries Mathieu Mari
  \footnote{Partially supported by the ERC CoG grant TUgbOAT no. 772346.}}
  \\\sffamily University of Warsaw\\\sffamily IDEAS-NCBR
  \and
  {\sffamily \bfseries Timothé Picavet}\\\sffamily ENS de Lyon\\\sffamily Aalto University
  \and 
  {\sffamily \bfseries Michał Pilipczuk
  \footnote{
    This work is a part of project BOBR that has received funding from the European Research Council (ERC) under the European Union’s Horizon 2020 research and innovation programme (grant agreement No. 948057).
  }}
  \\\sffamily University of Warsaw
}
\date{}

\begin{document}

\setlength{\abovedisplayskip}{3pt}
\setlength{\belowdisplayskip}{3pt}
\setlength{\abovedisplayshortskip}{0pt}
\setlength{\belowdisplayshortskip}{0pt}
\maketitle

\begin{tikzpicture}[remember picture, overlay]
\node[inner sep=0pt, above left = 0.1\linewidth] at (current page.south east) 
    {\includegraphics[height=0.1\textwidth]{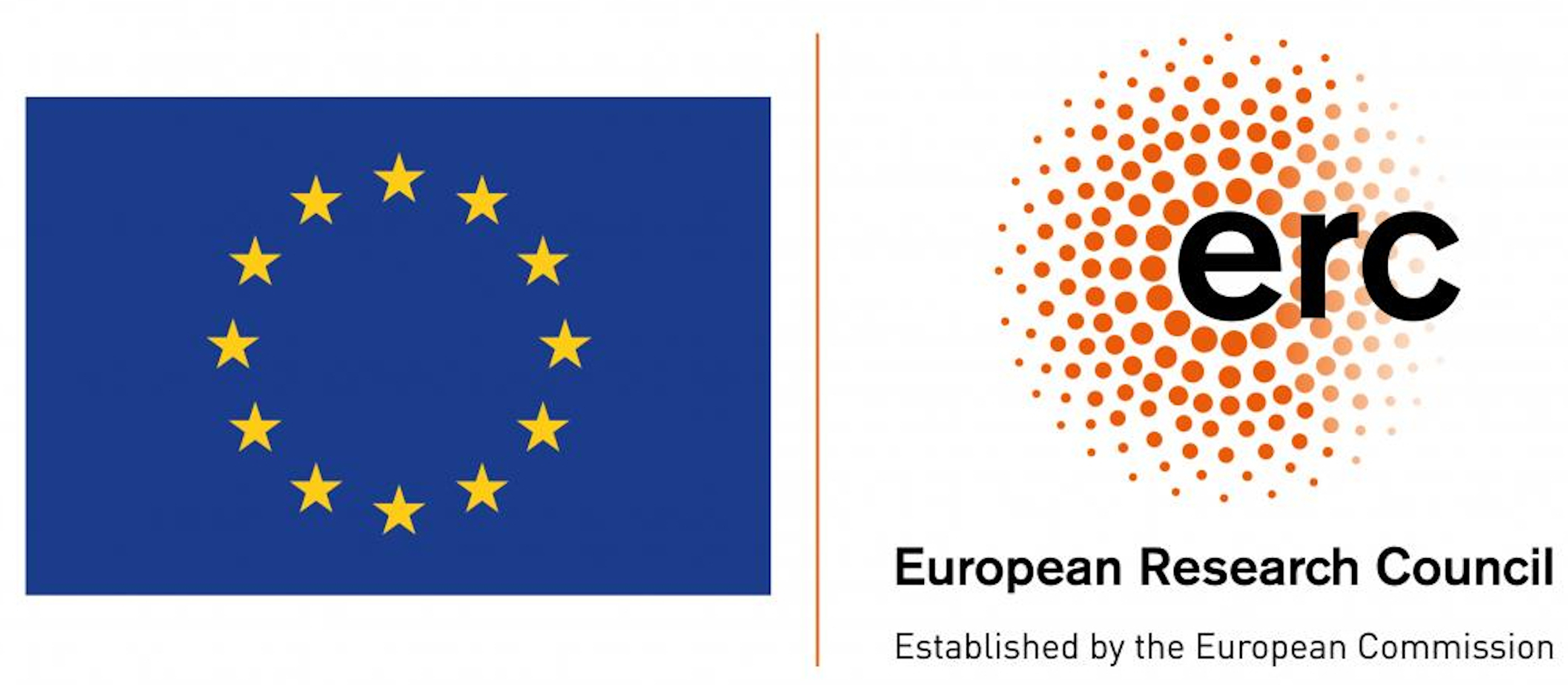}};
\end{tikzpicture}

\renewcommand{\abstractname}{\Large\sffamily Abstract}

\begin{abstract}
    We study a natural geometric variant of the classic {\sc{Knapsack}} problem called {\sc{2D-Knapsack}}: we are given a set of axis-parallel rectangles and a rectangular bounding box, and the goal is to pack as many of these rectangles inside the box without overlap. Naturally, this problem is $\mathsf{NP}$-complete. Recently, Grandoni et al. [ESA'19] showed that it is also $\mathsf{W}[1]$-hard when parameterized by the size $k$ of the sought packing, and they presented a parameterized approximation scheme (PAS) for the variant where we are allowed to rotate the rectangles by 90° before packing them into the box. Obtaining a PAS for the original {\sc{2D-Knapsack}} problem, without rotation, appears to be a challenging open question. 

    In this work, we make progress towards this goal by showing a PAS under the following assumptions:
    \begin{itemize}[nosep]
    \item both the box and  all the input rectangles
    have integral, polynomially bounded sidelengths;
    \item every input rectangle is wide --- its width is greater than its height; and
    \item the aspect ratio of the box is bounded by a constant.
    \end{itemize}
    Our approximation scheme relies on a mix of various parameterized and approximation techniques, including color coding, rounding, and searching for a structured near-optimum packing using dynamic programming.
\end{abstract}

\section{Introduction}

We study a natural geometric variant of the classic {\sc{Knapsack}} problem, called {\sc{2D Knapsack}} and defined as follows. On input, we are given a rectangular box $B$ and a set $\R$ of items, each being a rectangle. The task is to place as many items from $\R$ as possible in $B$ so that the placed items do not overlap. Note that this problem generalizes classic {\sc{Knapsack}}: given an instance of {\sc{Knapsack}} with items of sizes $a_1,\ldots,a_n$ and a knapsack of size $K$, we can create an instance of {\sc{2D Knapsack}} with $B$ being the $K\times 1$ rectangle and $\R$ consisting of $a_i\times 1$ rectangles, for all $i\in \{1,\ldots,n\}$.

As in the case of {\sc{Knapsack}}, there are two natural variants of the problem depending on how the input is encoded. In the {\em{binary variant}}, both $B$ and the rectangles of $\R$ have integral sidelengths encoded in binary, hence these sidelengths can be exponential in the total input size. In the {\em{unary variant}} the difference is that the sidelengths are encoded in unary, or equivalently, one assumes that all the sidelengths are bounded polynomially in the total input size. In this work we focus on the unary variant.

Again as in the case {\sc{Knapsack}}, adopting the unary variant helps tremendously for the design of algorithms for {\sc{2D Knapsack}}, for instance due to allowing to perform dynamic programming over the dimensions of the box. While the problem remains $\mathsf{NP}$-hard even in the unary variant~\cite{LeungTWYC90}, Adamaszek and Wiese~\cite{AdamaszekW15} gave a QPTAS for the problem in this setting.  The best approximation factor known to be achievable in polynomial time in the unary variant is $4/3+\eps$ due to Galvez et al.~\cite{Galvez00RW21}; earlier, a $(2+\eps)$-approximation was given in~\cite{JansenZ04} and a $(558/325 + \eps)$-approximation was given in~\cite{Galvez+21}. It is believed that the problem should admit a PTAS, but this question remains widely open to this day. 

We remark that the abovementioned works also study the weighted variant of the problem. In this work we only consider the unweighted version, hence an interested reader is invited to the relevant discussion in the references.

Recently, Grandoni et al.~\cite{grandoni2019pas} proposed to approach the question about the existence of a PTAS for {\sc{2D Knapsack}} by adding parameterization by the solution size to the picture. That is, they presented a {\em{parameterized approximation scheme}} ({\em{PAS}}) with running time of the form $k^{\O(k/\eps)}\cdot n^{\O(1/\eps^3)}$ that either finds a packing of size at least $(1-\eps)k$ or correctly concludes that there is no packing of size $k$. However, this result applies {\em{only}} to the variant of the problem where each input rectangle can be rotated by $90^\circ$ before packing it into the box, and the question about the existence of a PAS for {\sc{2D Knapsack}} without rotation was explicitly left open by Grandoni et al. This is in contrast with the other mentioned works on {\sc{2D Knapsack}} which all apply both to the variant with rotation and without rotation (for the variant with rotation, Galvez et al.~\cite{Galvez00RW21} reported a better approximation ratio of $5/4+\eps$).

We note that the PAS of Grandoni et al. actually works in the binary variant of the problem. Also, reliance on approximation is probably necessary: as proved in~\cite{grandoni2019pas}, the exact version of the problem is $\mathsf{W}[1]$-hard when parameterized by $k$.

\subparagraph*{Our contribution.} In this work we approach --- though not completely solve --- the open problem left by Grandoni et al.~\cite{grandoni2019pas} by giving a parameterized approximation scheme with running time of the form $f(k,\eps,\delta)\cdot n^{g(\eps)}$ for {\sc{2D Knapsack}} under the following assumptions:
\begin{itemize}[nosep]
    \item we consider the unary variant of the problem, thus the dimensions of the box are bounded polynomially in $n$;
    \item we assume that every item is {\em{wide}}: its width is not smaller than its height; and
    \item we assume that the {\em{aspect ratio}} (ratio between the dimensions) of the box is at most~$\delta$.
\end{itemize}
See \cref{thm:main} for a formal statement of our result and an explicit formula for the running time.
Note that in the context of the variant with rotation, the second assumption can be always achieved by rotating every input rectangle so that it is wide, while the third assumption for $\delta=1$ can be obtained by scaling both the box and all rectangles on input. 

Let us elaborate on our approach and how it is different from the approach of Grandoni et al.~\cite{grandoni2019pas}. The approach of~\cite{grandoni2019pas} can be summarized as follows. 
\begin{itemize}[nosep]
    \item Consider a hypothetical packing $\S$ of size $k$.
    \item {\em{Freeing a strip:}} Remove a small fraction of $\S$ and shift the items slightly in order to free up a horizontal strip of height $N/k^b$ at the bottom of the box, where $N$ is the sidelength of the box and $b=\O(1/\eps)$ is an integer. This strip can now accommodate all {\em{thin}} items: those of height at most $N/k^{b+1}$.
    \item {\em{Resource augmentation:}} After the previous step, we may assume that all items are {\em{large}} --- both dimensions are at least $N/k^{b+1}$ --- and there is still a considerable strip free at the bottom of the box. Now, one can round the heights of items up to the nearest multiplicity of, say, $N/k^{b+2}$ and argue that even the rounded items can be packed, due to the free strip at the bottom of the box. After rounding the rectangles have at most $k^{\O(1)}$ different heights, so keeping $k$ narrowest rectangles of each possible height gives us a polynomial in $k$ number of candidate rectangles that can be reasonably used in the packing. This easily leads to a PAS.  
\end{itemize}
The possibility of rotating rectangles is crucially used in the second step, freeing a strip. Without this assumption, thin rectangles come in two different flavors: there are {\em{wide}} rectangles, of very small height and possibly large width, and symmetric {\em{tall}} rectangles. The strategy of freeing a strip presented in~\cite{grandoni2019pas} can be applied also in the setting without rotation, but then it results in either freeing a horizontal strip at the bottom of the box, or a vertical strip at the left side of the box; there is no control over which strip will be freed. Consequently, only one type of thin rectangles can be disposed of as a result of freeing the strip, and there is no control over which one it is. 

In the setting of Grandoni et al., our assumptions on the problem essentially mean that we allow the existence of wide rectangles, but not of the tall ones. The application of the approach of Grandoni et al. could result in freeing a vertical strip, in which the wide rectangles cannot fit. Consequently, we do not see how to fix the approach presented in~\cite{grandoni2019pas} to solve our case where only wide rectangles are present and no rotation is allowed. We therefore abandon this approach and propose a completely new one.

Instead, we prove a different result about the existence of a well-structured near-optimum solution. Our structural lemma (\Cref{lem:structural}) says that at the cost of sacrificing a small fraction of rectangles, the considered packing can be divided into regions $B_1,B_2,\ldots,B_m$ so that:
\begin{itemize}
    \item Every region $B_i$ is delimited by the left side of the box, the right side of the box, and two $x$-monotone axis-parallel polylines connecting the left and the right side. Moreover, each of the polylines defining the division  $B_1,B_2,\ldots,B_m$ consists of $\O(1/\eps)$ segments.
    \item Every region $B_i$ is either {\em{light}} --- contains only $\O(1/\eps^2)$ rectangles from the packing --- or {\em{roundable}} --- rectangles within $B_i$ could be packed inside $B_i$ even after rounding them to the nearest multiple of (roughly) $N_1/k^2$, where $N_1$ is the width of the box.
\end{itemize}
Having such a structural lemma, a near-optimum solution can be constructed using a bottom-up dynamic programming that guesses the regions $B_i$ one by one. For each region $B_i$ we consider two cases: either $B_i$ is light and a solution within it can be guessed (essentially) by brute-force, or $B_i$ is roundable and using the same trick as in~\cite{grandoni2019pas}, one can restrict attention to $k^{\O(1)}$ many different candidates for rectangles that will be packed into $B_i$.

There is a technical caveat in the plan presented above. Namely, in dynamic programming we need to make sure that we do not reuse the same rectangle from $\R$ in two or more different regions $B_i$. We resolve this issue using color-coding. Namely, by applying color-coding upfront we may assume that all the rectangles in $\R$ are colored with $k$ colors, and we look for a packing consisting of rectangles of pairwise different colors. Then our dynamic programming keeps track of the subset of colors that have already been used, which adds only another dimension of size $2^k$ to the dynamic programming table.

\section{Preliminaries}

\subparagraph*{Basic terminology.}
For a positive integer $N$, we write $[N] = \{1,2,\dots, N\}$.

A {\em{rectangle}} is a pair of positive integers $R=(w, h)\in \mathbb{Z}_{+}^2$, and a {\em{placed rectangle}} is a set of the form $Q=[x,x+w]\times[y,y+h]\subseteq \mathbb{R}^2$, where $R=(w,h)$ is a rectangle and $(x,y)\in \mathbb{Z}^2$ is the bottom-left corner of $Q$; we will also say that such $Q$ is a {\em{placement}} of $R$. In the notation, we will sometimes treat placed rectangles as their non-placed counterparts; the meaning of this will be always clear from the context.

Both for placed and non-placed rectangles, $w$ and $h$ are  called the {\em{width}} and the {\em{height}}, respectively, and may be denoted by $w(P)$ and $h(P)$, where $P$ is the (placed) rectangle in question. The {\em{interior}} of a placed rectangle $Q=[x,x+w]\times[y,y+h]$ is the set $I(Q)=(x,x+w)\times(y,y+h)$. Two placed rectangles {\em{overlap}} if their interiors intersect. 

A {\em{zone}} is simply a subset of $\mathbb{R}^2$. For a zone $Z$ and a set of placed rectangles $\R$, by $\R[Z]= \{R \in \R \mid R\subseteq Z\}$ we denote the set of all rectangles from $\R$ that are entirely contained in $Z$.
For a zone $Z$ and a set of non-placed rectangles $\R$, a {\em{packing}} of $\R$ in $Z$ is a set $\R'=\{R'\colon R\in \R\}$ consisting of pairwise non-overlapping placed rectangles contained in $Z$, where $R'$ is a placement of $R$ for each $R\in \R$.

\subparagraph{The problem and the main result.}
In the (parameterized variant of) {\sc{2D Knapsack}} problem, we are given a rectangular zone $B=[0,N_1] \times [0,N_2]\subseteq \mathbb{R}^2$ called the {\em{box}}, where $N_1,N_2$ are positive integers, a set $\R$ of rectangles called {\em{items}}, and an integer $k$. The question is whether there exists a packing of some $k$ items from $\R$ in the box~$B$. 

In the context of an instance $(B,\R,k)$ of {\sc{2D Knapsack}},
the {\em{size}} of the box $B$ is $\|B\|=N_1+N_2$, and the {\em{aspect ratio}} of $B$ is $\delta(B)=\max\left(\frac{N_1}{N_2},\frac{N_2}{N_1}\right)$. Further, an item $R\in \R$ is {\em{wide}} if $w(R)\geq h(R)$. Note that in the variant of the problem where rotations by $90^\circ$ are allowed, one may always rotate the items so that they are wide. When the instance $(B,\R,k)$ is clear from the context, by a {\em{packing}} we mean a packing of a subset of $\R$ in $B$.

With these definitions in place, we can state our main result.
\begin{theorem}\label{thm:main}
    There exists an algorithm that given an accuracy parameter $\eps>0$ and an instance $(B,\R,k)$ of {\sc{2D Knapsack}}, where $\R$ consists only of wide items, either returns a packing of size at least $(1-\eps)k$ or correctly concludes that there is no packing of size $k$. The running time of the algorithm is $\delta(B)^{\O(k)}\cdot (k+1/\eps)^{\O(k+1/\eps^2)}\cdot (|\R|\|B\|)^{\O(1/\eps^2)}$.
\end{theorem}

\subparagraph*{Polylines and containers.}
In our algorithm for {\sc{2D Knapsack}} we will decompose the box into zones delimited by borders of low complexity, allowing those borders to be efficiently guessed.
Formally, each border will be a polyline defined as follows.
\begin{definition}[Axis-parallel polyline]
    An \emph{axis-parallel polyline} $\P$ is a union of horizontal or vertical segments $S_1, S_2, \dots, S_m$ such that for $1 \leq i \leq m -1$, the end of segment $S_i$ is the beginning of segment $S_{i+1}$.
    Then $m$ is called the {\em{complexity}} of $\P$.
\end{definition}
For brevity, axis-parallel polylines will be just called polylines.
We will only work with \emph{monotone} polylines, meaning that all horizontal coordinates of points on $S_j$ will not be smaller than the horizontal coordinates of the points on $S_i$, whenever $i<j$.
A polyline $P$ \emph{crosses} a placed rectangle $R$ if $P$ intersects the interior of $R$.

Next we introduce {\em{containers}}. We will use them to capture the idea of decomposing the box into~zones.

\begin{definition}[Container]\label{def:container}
    A \emph{container} $\C$ is a union of horizontal or vertical segments $S_1, S_2, \dots, S_m$ such that:
    \begin{itemize}
        \item for $1 \leq i \leq m -1$, the end of segment $S_i$ is the beginning of segment $S_{i+1}$, and
        \item the end of segment $S_m$ is the beginning of segment $S_1$.
    \end{itemize}
    Furthermore, we require that $\C$ is weakly-simple (as introduced in \cite{kusakari1999shortest,demaine2007geometric,chang2014detecting}) in the following sense: if $\gamma\colon S^1\to \mathbb{R}^2$ is a parameterization of $\C$, then for every $\eps > 0$ there exists an injective continuous $\gamma_\eps\colon S^1 \to \mathbb{R}^2$ such that $\|\gamma - \gamma_\eps\|_\infty \leq \eps$.     

    The {\em{inside}} of the container, denoted $I(\C)$, is the bounded open region delimited by the segments.
    Moreover, the {\em{complexity}} of the container is defined as $m$.
\end{definition}
Note that the inside of a container is not necessarily connected.
For clarification, see \Cref{fig:containers}.
\begin{figure}[H]
    \centering
    \begin{tikzpicture}[very thick, num/.style={fill=white,anchor=center,inner sep=2pt}]

        \path[-{Stealth[left]}]
            (0,1) edge node[style=num] {1} (2.5,1)
            (2.5,1) edge node[style=num,anchor=west] {2} (2.5,0)
            (2.5,0) edge node[style=num,anchor=south] {3} (5,0)
            (5,0) edge node[style=num,anchor=east] {4} (5,1)
            (5,1) edge node[style=num] {5} (6,1)
            (6,1) edge node[style=num] {6} (6,-1)
            (6,-1) edge node[style=num] {7} (4.5,-1)
            (4.5,-1) edge node[style=num,anchor=east] {8} (4.5,0)
            (4.5,0) edge node[style=num,anchor=north] {9} (1,0)
            (1,0) edge node[style=num,anchor=west] {10} (1,-1)
            (1,-1) edge node[style=num,anchor=north] {11} (0,-1)
            (0,-1) edge node[style=num] {12} (0,1);
        \fill[fill=Gray,opacity=0.3] (0,1) -- (2.5,1) -- (2.5,0) -- (5,0) -- (5,1) -- (6,1) -- (6,-1) -- (4.5,-1) -- (4.5,0) -- (1,0) -- (1,-1) -- (0,-1) -- cycle;

        \path[draw=NavyBlue,opacity=0.3]
            (0,1) edge[out=60,in=120] (2.5,1)
            (2.5,1) edge[out=300,in=180] (3.5,0.2)
            (3.5,0.2) edge[out=360,in=240] (5,1)
            (4.5,-1) edge[out=300,in=240] (6,-1)
            (6,-1) edge[out=60,in=300] (6,1)
            (6,1) edge[out=120,in=60] (5,1)
            (4.5,-1) edge[out=120,in=360] (3.5,-0.2)
            (3.5,-0.2) edge[out=180,in=60] (1,-1)
            (1,-1) edge[out=240,in=300] (0,-1)
            (0,-1) edge[out=120,in=240] (0,1);
        
        \path[-{Stealth[left]}]
            (8,1) edge node[style=num] {1} (10.5,1)
            (10.5,1) edge node[style=num,anchor=west] {2} (10.5,0)
            (10.5,0) edge node[style=num,anchor=south] {3} (12.5,0)
            (12.5,0) edge node[style=num,anchor=east] {4} (12.5,-1)
            (12.5,-1) edge node[style=num] {5} (14,-1)
            (14,-1) edge node[style=num] {6} (14,1)
            (14,1) edge node[style=num] {7} (13,1)
            (13,1) edge node[style=num,anchor=east] {8} (13,0)
            (13,0) edge node[style=num,anchor=north] {9} (9,0)
            (9,0) edge node[style=num,anchor=west] {10} (9,-1)
            (9,-1) edge node[style=num,anchor=north] {11} (8,-1)
            (8,-1) edge node[style=num] {12} (8,1);

        \path[draw=RedOrange,opacity=0.3]
            (8,1) edge[out=60,in=120] (10.5,1)
            (10.5,1) edge[out=300,in=160] (11.5,0)
            (11.5,0) edge[out=340,in=120] (12.5,-1)
            (12.5,-1) edge[out=300,in=240] (14,-1)
            (14,-1) edge[out=60,in=300] (14,1)
            (14,1) edge[out=120,in=60] (13,1)
            (13,1) edge[out=240,in=20] (11.5,0)
            (11.5,0) edge[out=200,in=60] (9,-1)
            (9,-1) edge[out=240,in=300] (8,-1)
            (8,-1) edge[out=120,in=240] (8,1);

    \end{tikzpicture}
    \caption{Left panel: an example of a container where the order of the segments is given by the numbers, and some injective candidate $\gamma_\eps$ in light blue. Right panel: an example of a non-container (the paths cross in the middle), and some candidate $\gamma_\eps$ in light red, that is not injective.}
    \label{fig:containers}
\end{figure}
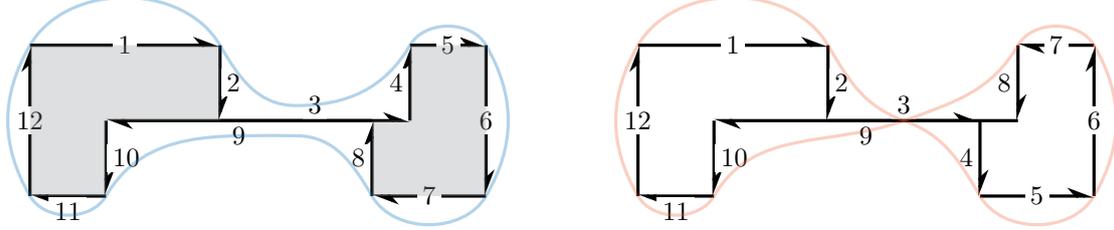

\section{Exact algorithm}

In this section, we give an exact algorithm for the problem, which will be later used as a subroutine in the proof of \Cref{thm:main}. The point here is that we allow the box to be delimited by an arbitrary container, and we measure the running time in the complexity of the container. Formally, we will prove the following statement.

\begin{lemma}\label{lem:algo_bounded_rects}
    Given a set of rectangles $\R$ and a container $\C$ of complexity $m$, one can determine whether there is a packing of the rectangles of $\R$ inside $\C$ in time $(m+|\R|)^{\O(|\R|)}$.
\end{lemma}

The main idea of our algorithm is to push the packing bottom-left, as explained in the next definition.

\begin{definition}
    A packing $\R$ inside a container $\C$ is said to be \emph{pushed bottom-left} if for every rectangle $R\in \R$, its left (resp. bottom) side intersects either a vertical (resp. horizontal) segment of the container, or a right (resp. bottom) side of another rectangle $R'\in\R$.
\end{definition}

It is not hard to see that if a packing is pushed bottom-left, then there must be a rectangle in the packing whose left and bottom sides rest on the perimeter of the container. This is formally proved in the following statement.

\begin{proposition}\label{prop:bottom_left_rect}
    Suppose $\R$ is a non-empty packing of rectangles inside a container $\C$ that is pushed bottom-left.
    Then there exists a rectangle $R\in\R$ such that its left side intersects a vertical segment of the container and its bottom side intersects an horizontal segment of the container.
\end{proposition} 
\begin{proof}
    Create a directed graph $D$ with vertex set $\R$, where there is an edge $(R,R')$ if the bottom side of $R$ intersects the top side of $R'$ on more than a single point, or if the left side of $R$ intersects the right side of $R'$ on more than a single point.
    Let us show that $D$ has no directed cycle, so for contradiction suppose $R_1,\ldots,R_\ell$ is a directed cycle in $D$. In what follows, all indices behave cyclically modulo $\ell$.

    For each $i\in [\ell]$, select an arbitrary point $p_i$ in the intersection of $R_{i-1}$ and $R_i$ that is neither a corner of $R_{i-1}$ nor a corner of $R_i$. Further, observe that one can construct a curve $\gamma_i\colon [0,d_i] \to R_i$, where $d_i$ is the length of $\gamma_i$, such that:
    \begin{itemize}
        \item $\gamma_{i}$ is smooth (formally, $C^1$) and monotone in both directions,
        \item $\|\gamma_{i}'(t)\|_2=1$ for all $t\in [0,d_i]$,
        \item $\gamma_{i}(0)=p_i$ and $\gamma_{i}(d_i)=p_{i+1}$, and
        \item the tangent of $\gamma_{i}$ at $p_i$ and $p_{i-1}$ is perpendicular to the respective side and faces the inside (resp. outside) of $R_i$. For instance if $p_i$ is on the top side of $R_i$, we require $\gamma'_{i}(0)=(0,-1)$, and if $p_{i+1}$ is on the left side of $R_i$, we require $\gamma'_{i}(d_i)=(-1,0)$.
    \end{itemize}
    An example of such a construction is shown below.
    \begin{center}    
    \begin{tikzpicture}[scale=0.11,
            rect/.style={very thick,draw=Black},
            curve/.style={very thick,color=NavyBlue,postaction={decorate}},
            point/.style={circle,fill=NavyBlue,inner sep=2pt},
            decoration={markings, mark= at position 0.5 with {\arrow{Latex[round]}}}
        ] 
        \draw[style=rect] (0,0) rectangle (20,10);
        \node[style=point,label=above:\color{NavyBlue}$a$] (a) at (14,10) {};
        \node[style=point,label=below left:\color{NavyBlue}$b$] (b) at (6,0) {};
        \draw[style=curve] (a) to[out=270,in=90] (b);
        \draw[style=rect] (25,0) rectangle (45,10);
        \node[style=point,label=above:\color{NavyBlue}$a$] (c) at (35,10) {};
        \node[style=point,label=below left:\color{NavyBlue}$b$] (d) at (25,5) {};
        \draw[style=curve] (c) to[out=270,in=0] (d);
        \draw[style=rect] (50,0) rectangle (70,10);
        \node[style=point,label=right:\color{NavyBlue}$a$] (e) at (70,8) {};
        \node[style=point,label=above left:\color{NavyBlue}$b$] (f) at (50,2) {};
        \draw[style=curve] (e) to[out=180,in=0] (f);
        \draw[style=rect] (75,0) rectangle (95,10);
        \node[style=point,label=right:\color{NavyBlue}$a$] (g) at (95,5) {};
        \node[style=point,label=below left:\color{NavyBlue}$b$] (h) at (85,0) {};
        \draw[style=curve] (g) to[out=180,in=90] (h);
        \draw[->] (a) -- ++(0,-4);
        \draw[->] (c) -- ++(0,-4);
        \draw[->] (d) -- ++(-4,0);
        \draw[->] (e) -- ++(-4,0);
        \draw[->] (b) -- ++(0,-4);
        \draw[->] (f) -- ++(-4,0);
        \draw[->] (g) -- ++(-4,0);
        \draw[->] (h) -- ++(0,-4);
        \pgfdeclarepatternformonly{swnestripes2}{\pgfpoint{0cm}{0cm}}{\pgfpoint{1cm}{1cm}}{\pgfpoint{1cm}{1cm}}
        {
            \foreach \i in {0.1cm, 0.3cm,...,0.9cm}
            {
             \pgfpathmoveto{\pgfpoint{\i}{0cm}}
             \pgfpathlineto{\pgfpoint{1cm}{1cm - \i}}
             \pgfpathlineto{\pgfpoint{1cm}{1cm - \i + 0.1cm}}
             \pgfpathlineto{\pgfpoint{\i - 0.1cm}{0cm}}
             \pgfpathclose%
             \pgfusepath{fill}
             \pgfpathmoveto{\pgfpoint{0cm}{\i}}
             \pgfpathlineto{\pgfpoint{1cm - \i}{1cm}}
             \pgfpathlineto{\pgfpoint{1cm - \i - 0.1cm}{1cm}}
             \pgfpathlineto{\pgfpoint{0cm}{\i + 0.1cm}}
             \pgfpathclose%
             \pgfusepath{fill}
            }
        }
        \begin{scope}
            \clip (110,5) circle (10cm);
            \pattern[pattern=swnestripes2, pattern color=RedOrange!50] (110,5) rectangle (120,15);
        \end{scope}
        \draw[very thick,NavyBlue] (110,5) circle (10cm);
        \foreach \x in {90,100,...,360}
            \draw[->] (110,5) -- ({110+10*sin(\x)},{5+10*cos(\x)});
    \end{tikzpicture}
    \end{center}

    Concatenating all the curves $\gamma_i$ in order yields a smooth closed curve $\gamma\colon S\to \bigcup_{i=1}^\ell R_i$ without self-crossings such that $\|\gamma'(t)\|_2=1$ for all $t\in S$, where $S$ is the circle of length $\sum_{i=1}^\ell d_i$. Here is the crucial observation: by the way we oriented the arcs in $D$, the vector $\gamma'(t)$ is never in the positive orthant (i.e. $\gamma'(t)$ has not both coordinates positive), for any $t\in S$.
    
    \begin{center}    
    \begin{tikzpicture}[
            scale=0.2,
            rect/.style={very thick,draw=Black},
            curve/.style={very thick,color=NavyBlue,postaction={decorate}},
            decoration={markings, mark= at position 0.5 with {\arrow{Latex[round]}}}
        ]
        \node (0) at (-20,25) {};
        \draw[style=rect] (-35,20) rectangle (-15,25);
        \node (1) at (-17.5,20) {};
        \draw[style=rect] (-20,15) rectangle (-5,20);
        \node (2) at (-7.5,15) {};
        \draw[style=rect] (-10,10) rectangle (5,15);
        \node (3) at (2.5,10) {};
        \node (C) at (3,5) {\color{NavyBlue}$\gamma$};
        \draw[style=rect] (0,0) rectangle (20,10);
        \node (4) at (0,5) {};
        \draw[style=rect] (-25,8) rectangle (0,4);
        \node (5) at (-25,7.5) {};
        \draw[style=rect] (-40,15) rectangle (-25,5);
        \node (6) at (-30,5) {};
        \draw[style=curve] (0.center) to[out=270,in=90] (1.center);
        \draw[style=curve] (1.center) to[out=270,in=90] (2.center);
        \draw[style=curve] (2.center) to[out=270,in=90] (3.center);
        \draw[style=curve] (3.center) to[out=270,in=0] (4.center);
        \draw[style=curve] (4.center) to[out=180,in=0] (5.center);
        \draw[style=curve] (5.center) to[out=180,in=90] (6.center);
        % \draw[style=graph] (0,20) to[out=180,in=0] (20,30);
    \end{tikzpicture}
    \end{center}
    
    However by Theorem 2 of \cite[section 5-7, page 402]{do2016differential}, a smooth closed curve in the plane without self-crossings has rotation index $\pm 1$, where the rotation index of a curve is the number of times its tangent vector turns around the origin.
    This means that by the intermediate value theorem, for every $\alpha\in]0,2\pi[$, there exists a point of the curve where the tangent vector is at angle $\alpha$ with the $x$-axis, and hence belongs to the positive orthant. This is a contradiction.

    We conclude that $D$ is acyclic, hence it has a sink $R$ --- a rectangle with out-degree $0$.
    Therefore, $R$ is the rectangle we want: its left side intersects a vertical segment of container and its bottom side intersects a horizontal segment of the container. 
\end{proof}

With \Cref{prop:bottom_left_rect} established, we can conclude our goal using a simple branching strategy.

\begin{proof}[Proof of \Cref{lem:algo_bounded_rects}]
    We prove a stronger statement where we allow $\C$ to be the union of several disjoint containers, and we let $m$ be the sum of their complexities.

    Suppose there is a packing packing $\S$ of the rectangles of $\R$ into $\C$.
    We can assume without loss of generality that $\S$ is pushed bottom-left within every container of $\C$.
    Now by \Cref{prop:bottom_left_rect}, there exists a rectangle $R$ such that its left (resp. bottom) side intersects a vertical (resp. horizontal) segment of a container in $\C$.
    
    So here is a recursive procedure to solve the problem.
    First, guess (by trying all possibilities) the rectangle $R$ satisfying the condition above; there are $n$ different possibilities for $R$, where $n=|\R|$.
    Second, guess which pair of segments of the containers intersect the left and the bottom side of $R$; there are at most $m^2$ possibilities.
    Place rectangle $R$ according to the latter guess and verify that it is indeed fully contained in $\C$. Then, ``carve out'' the rectangle, i.e., define a new union of containers $\C'$ so that the $I(\C') = I(\C)\setminus R$.

    It is easy to see that the total complexity of the new union of containers $\C'$ is at most $m + 6$ and it can be computed in time polynomial in $m$.
    Then, recurse on $\R'$ and $\C'$ where $\R'= \R\setminus \{R\}$.

    It is clear that the algorithm is correct. To analyze its running time, note that the recursion tree has depth bounded by $n$ and branching bounded by $n(m+6n)^2$, hence it consists of $(m+n)^{\O(n)}$ nodes. The internal computation at each node take time polynomial in $n$ and $m$, so the total running time of $(m+n)^{\O(n)}$ follows.
\end{proof}

\section{Giving structure to the packing}

In this section we prove structural results that can be summarized as follows: at the cost of sacrificing a small fraction of the packing, one can apply resource augmentation --- round the packing --- so that it gains a certain structure. Once this structure is achieved, we will argue later that structured packings can be efficiently computed using dynamic programming.

Throughout this section we fix an instance $(B,\R,k)$ of {\sc{2D Knapsack}}, where $B=[0,N_1]\times [0,N_2]$ and  $\R$ consists only of wide rectangles: $w(R)\geq h(R)$ for all $R\in \R$.

To perform resource augmentation, we need the following notion of rounding a rectangle. 
Informally, a rounded rectangle is the original rectangle with its width rounded up to the nearest multiple of $\ell'=\ell^2/N_1$; here is a formal definition.

\begin{definition}[Rounded rectangles]
    Let $R=(w,h)$ be a rectangle and $\ell>0$ be a positive real.
    Then the \emph{$\ell$-rounded}  rectangle $\roundell{R}$ is the rectangle $(\ell'\ceil{w/\ell'}, h)$ where $\ell'= \ell^2/N_1$. For a set $\R$ of rectangles, we define similarly $\roundell{\R} = \{\roundell{R} \colon R\in\R\}$.
\end{definition}

As mentioned, the key idea behind our algorithm is to look for a specifically structured packing. This structure is quantified formally in the following definition.
Broadly speaking, we look for a packing that is partitioned into regions of low complexity and such that the rectangles in each region behave well.

\begin{definition}[Structured packing]
Fix any $\eps,\ell>0$.
Consider a set of pairwise non-intersecting monotone polylines $P_1, P_2, \dots, P_m$ contained in the box $B$, where each $P_i$  starts at the left side of $B$ and finishes at the right side of $B$, and the polylines $P_1,\ldots,P_m$ are naturally numbered from bottom to top. We define the partition of the box $B$ into regions $B_0,B_1,\ldots,B_m$ so that each region $B_i$ is delimited by the polylines $P_i$ and $P_{i+1}$ and the left and the right side of $B$ (here we define for convenience $P_0$ to be the bottom side of $B$ and $P_{m+1}$ to be the top of $B$).

We say that a packing of rectangles $\Q$ in $B$ is an \emph{$(\eps,\ell)$-structured packing} if every rectangle in $\Q$ has width at least $2\ell$ and there exist polylines $P_1, P_2, \dots, P_m$ as above, each of complexity at most $4/\eps+1$, such that no rectangle of $\Q$ is crossed by any polyline $P_i$, $i\in[m]$, and for each $i$,$0 \leq i \leq m$ at least one of the following conditions holds:
\begin{itemize}
    \item $|\Q[B_i]| \leq 2/\eps^2$, or
    \item $\roundell{\Q[B_i]}$ can be packed into $B_i$.
\end{itemize}
\end{definition}

The rest of the section is dedicated to proving the following structural lemma (recall that the instance $(B,\R,k)$ is fixed in the context):
\begin{lemma}[structural lemma]\label{lem:structural}
    Suppose $\ell>0$ is a positive real such there is a packing of size $k$ consisting of rectangles of width at least $2\ell$ each.
    Then for every $\eps>0$, there exists also an $(\eps,\ell)$-structured packing of size at least $(1-3\eps)k$.
\end{lemma}

This section is divided into 3 parts. 
In the first subsection we study the assumed packing of size $k$ and define an associated {\em{conflict graph}}, which turns out to be planar.
In the second subsection, we show that if there exists a packing of rectangles in a specific zone on the box, then at the cost of removing a few rectangles, there exists a packing of the rounded rectangles into a slightly bigger rounded version of the zone.
In the last section, we define the specific polylines that will divide the zones and finish the proof of the \Cref{lem:structural}.

By assumption, there exists a packing $\S$ in $B$ consisting of $k$ rectangles from $\R$, each of width at least~$2\ell$. Fix $\S$ for the remainder of this section.

\subsection{Conflict graph}

For the definition of the conflict graph, we need the following notion of horizontal visibility.

\begin{definition}
    Two different placed rectangles $R,R'\in \S$ \emph{see each other} if there is an horizontal segment $s$ intersecting the interior of the right side of $R$ and the interior of the left side of $R'$ (or vice versa) such that $s$ does not intersect any other rectangle of $\S$. Notice that $s$ may consist of a single point, if $R$ and $R'$ are touching. For convenience, we extend this definition to the case where $R$ is the left side of $B$ or $R'$ is the right side of $B$.
    For instance with the left side we associate the placed rectangle $R_{\textrm{left}}=[-1,0]\times [0,N_2]$ and say that $R$ and the left side see each other if $R_{\textrm{left}}$ and $R$ see each other; similarly for the right side.
    The left side and the right side \emph{do not} see each other.
\end{definition}

Note two rectangles intersecting only at their common corner do {\em{not}} see each other.

\begin{definition}[Conflict graph]
    For a packing $\S$, we define the \emph{conflict graph} of $\S$ to be the graph $G$ defined as follows: the vertex set contains all the rectangles of $\S$, and in addition there are two special vertices $s$ and $t$ identified with the left side and the right side of $B$, respectively. Two vertices of $G$ are adjacent if and only if they see each other.
\end{definition}

It is easy to see that the conflict graph is planar; see \Cref{fig:conflict}.
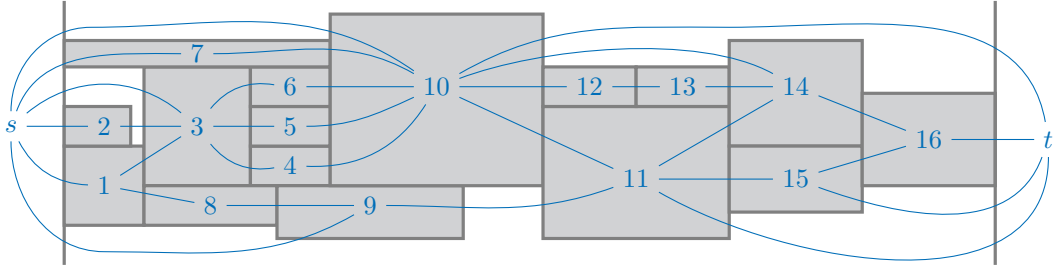
\begin{figure}[H]
    \centering
    \begin{tikzpicture}[scale=0.175,
        rect/.style={very thick,draw=black!50,fill=Gray!40},graph/.style={color=NavyBlue}]
        \draw[style=rect] (0,0) -- (0,20);
        \draw[style=rect] (0,3) rectangle (6,9);
        \node[style=graph] (1) at (3,6) {1};
        \draw[style=rect] (0,9) rectangle (5,12);
        \node[style=graph] (2) at (3,10.5) {2};
        \draw[style=rect] (6,6) rectangle (14,15);
        \node[style=graph] (3) at (10,10.5) {3};
        \draw[style=rect] (14,6) rectangle (20,9);
        \node[style=graph] (4) at (17,7.5) {4};
        \draw[style=rect] (14,9) rectangle (20,12);
        \node[style=graph] (5) at (17,10.5) {5};
        \draw[style=rect] (14,12) rectangle (20,15);
        \node[style=graph] (6) at (17,13.5) {6};
        \draw[style=rect] (0,15) rectangle (20,17);
        \node[style=graph] (7) at (10,16) {7};
        \draw[style=rect] (6,6) rectangle (16,3);
        \node[style=graph] (8) at (11,4.5) {8};
        \draw[style=rect] (16,6) rectangle (30,2);
        \node[style=graph] (9) at (23,4.5) {9};
        \draw[style=rect] (20,6) rectangle (36,19);
        \node[style=graph] (10) at (28,13.5) {10};
        \draw[style=rect] (36,2) rectangle (50,12);
        \node[style=graph] (11) at (43,6.5) {11};
        \draw[style=rect] (36,12) rectangle (43,15);
        \node[style=graph] (12) at (39.5,13.5) {12};
        \draw[style=rect] (43,12) rectangle (50,15);
        \node[style=graph] (13) at (46.5,13.5) {13};
        \draw[style=rect] (50,9) rectangle (60,17);
        \node[style=graph] (14) at (55,13.5) {14};
        \draw[style=rect] (50,9) rectangle (60,4);
        \node[style=graph] (15) at (55,6.5) {15};
        \draw[style=rect] (60,6) rectangle (70,13);
        \node[style=graph] (16) at (65,9.5) {16};
        \draw[style=rect] (70,0) -- (70,20);
        \draw[style=graph] (11) to (14);
        \node[style=graph] (s) at (-4,10.5) {$s$};
        \node[style=graph] (t) at (74,9.5) {$t$};
        \draw[style=graph] (s) to[bend right] (1);
        \draw[style=graph] (s) to[out=270,in=180] (3,1) to[out=0,in=210] (9);
        \draw[style=graph] (s) to[out=90,in=180] (2,18) to[out=0,in=145] (10);
        \draw[style=graph] (s) to (2);
        \draw[style=graph] (s) to[out=70,in=180] (7);
        \draw[style=graph] (s) to[out=40,in=140] (3);
        \draw[style=graph] (2) to (3);
        \draw[style=graph] (1) to (3);
        \draw[style=graph] (1) to (8);
        \draw[style=graph] (3) to[bend left] (6);
        \draw[style=graph] (3) to (5);
        \draw[style=graph] (3) to[bend right] (4);
        \draw[style=graph] (8) to (9);
        \draw[style=graph] (7) to[out=0,in=150] (10);
        \draw[style=graph] (6) to (10);
        \draw[style=graph] (5) to[out=0,in=205] (10);
        \draw[style=graph] (4) to[bend right] (10);
        \draw[style=graph] (10) to (12);
        \draw[style=graph] (10) to (11);
        \draw[style=graph] (9) to[out=0,in=200] (11);
        \draw[style=graph] (10) to[bend left,looseness=0.7,out=20] (14);
        \draw[style=graph] (12) to (13);
        \draw[style=graph] (13) to (14);
        \draw[style=graph] (11) to (15);
        \draw[style=graph] (14) to (16);
        \draw[style=graph] (15) to (16);
        \draw[style=graph] (16) to (t);
        \draw[style=graph] (15) to[out=340,in=250] (t);
        \draw[style=graph] (11) to[out=330,in=270] (t);
        \draw[style=graph] (10) to[out=30,in=180] (50,18) to[out=0,in=100] (t);
    \end{tikzpicture}
    \caption{Example of a conflict graph.}\label{fig:conflict}
\end{figure}
We formalize this intuition in the following lemma.
\begin{lemma}\label{lem:planarity}
    For any packing $\S$, the conflict graph of $\S$ is planar.
\end{lemma}
\begin{proof}
    Let $G$ be the conflict graph of $\S$.
    For a rectangle $R\in\S$, we denote its associated vertex in $G$ by $v_R$.
    We define a planar embedding of $G$ as follows.
    We define the position of a vertex $v_R$ to be the center $c(R)=(x(R)+w(R)/2, y(R)+h(R)/2)$ of the corresponding rectangle.
    If there is an edge $e=v_Rv_{R'}$, choose $y\in \mathbb{R}$ such that the horizontal segment $s = [x(R)+w(R),x(R')] \times \{y\}$ witnesses that $R$ and $R'$ that see each other, where we assume w.l.o.g. that $x(R)+w(R)\leq x(R')$. We define the embedding $\gamma_e$ of $e$ as the union of 3 internally disjoint segments:
    \begin{itemize}[nosep]
        \item $s^1_e = [c(R),(x(R)+w(R), y))]$,
        \item $s^2_e = s$,
        \item $s^3_e = [(x(R'), y),c(R')]$.
    \end{itemize}
    It is straightforward to check that all the curves $\gamma_e$ $e\in E(G)$ are pairwise internally disjoint, hence they constitute a planar embedding of $G$.
\end{proof}

\subsection{Packing rounded rectangles}

Next, we analyze a packing within some zone $Z\subseteq \mathbb{R}^2$, with the goal of understanding when and how the rectangles of this packing can be rounded to obtain a rounded packing of substantial size. 
We fix some positive real $\ell>0$ for the rest of this subsection.

First, we need some definitions about expanding zones. 
\begin{definition}
    Let $Z\subseteq \mathbb{R}^2$. We define:
    \begin{itemize}
        \item the \emph{negatively shifted zone} $\overleftarrow{Z}\!\langle \ell \rangle =  \left(\bigcup_{(x,y) \in Z} [x-\ell, x]\times \{y\}\right) \cap [0,N_1-\ell]\times [0,N_2]$,
        \item the \emph{positively shifted zone} $\overrightarrow{Z}\!\langle \ell \rangle =  \left(\bigcup_{(x,y) \in Z} [x, x+\ell]\times \{y\}\right) \cap [0,N_1]\times [0,N_2]$,
        \item and the \emph{rounded zone} $\overleftrightarrow{Z}\!\langle \ell \rangle = \textstyle\left(\bigcup_{(x,y) \in Z} [x-\ell, x+\ell]\times \{y\}\right) \cap [0,N_1]\times [0,N_2].$
    \end{itemize} 
\end{definition}
Note that if $Z' = \overleftarrow{Z}\!\langle \ell \rangle$ then $\overleftrightarrow{Z}\!\langle \ell \rangle = \overrightarrow{Z'}\!\langle \ell \rangle$, and that the first two definitions are not symmetric.

Our main goal in this subsection is to prove the following lemma.
It intuitively says that at the cost of removing an $st$-separator in the conflict graph, one can find a packing of the rounded rectangles into a slightly extended zone. Here, an {\em{$st$-separator}} is a set of vertices (rectangles) that hits every $s$-$t$ path.

\begin{lemma}\label{lem:approx_widths}
    Let $\Q$ be a packing in a zone $Z\subseteq B$ such that every rectangle of $\Q$ has width at least $2\ell$. Further, let $C$ be an $st$-separator in the conflict graph of $\Q$.
    Then $\roundell{\Q \setminus C}$ can be packed inside the zone $\overleftrightarrow{Z}\!\langle \ell \rangle$.
\end{lemma}

The first step towards the proof of \Cref{lem:approx_widths} is to repack $\Q$ into the negatively shifted zone $Z$ at the cost of deleting a few rectangles.

\begin{proposition}\label{prop:z-}
    Let $\Q$ be a packing in a zone $Z\subseteq B$ such that every rectangle of $\Q$ has width at least $2\ell$. Further, let $C$ be an $st$-separator in the conflict graph of $\Q$.
    Then $\Q \setminus C$ can be packed inside the zone $\overleftarrow{Z}\!\langle \ell \rangle$.
\end{proposition}
For an illustration of the proof, see \Cref{fig:z-}.
\begin{proof}
    Let $G$ be the conflict graph of $\Q$. 
    Since $C$ is an $st$-separator in $G$, we may partition $\Q$ into three disjoint sets $X$, $C$, and $Y$ so that vertices of $X$ are not connected to $t$, vertices of $Y$ are not connected to $s$, and no vertex of $X$ is adjacent to any vertex of $Y$.
    Now, construct a new set of placed rectangles $\Q'$ by removing all rectangles of $C$ and shifting every rectangle of $Y$ by $\ell$ to the left. It remains to prove that $\Q'$ is a packing and that all rectangles of $\Q'$ are entirely contained in $\overleftarrow{Z}\!\langle \ell \rangle$.
    
    For the second assertion, we need to prove that (i) no $R'\in \Q'$ crosses the left side of the box, i.e., no $R'\in \Q'$ is such that $x(R') < 0$, and (ii) no rectangle $R'\in \Q'$ contains a point with horizontal coordinate larger than $N_2-\ell$, i.e. $x(R')+w(R')>N_2-\ell$. 
    To prove (i), suppose for the sake of a contradiction that there exists $R'\in Y$ such that $x(R')<0$. 
    We must have $R'\in Y$ because $R'$ was shifted, and hence $R'$ cannot see the left side of the box.
    Let $R\in\Q$ be $R'$ before shifting. We know that $x(R)<\ell$, therefore as every rectangle has width at least $\ell$, there is no rectangle in $\Q$ that would be placed between $R$ and the left side of the box. Therefore, the $R$ must see the left side of the box, which is a contradiction because $R\in Y$. A symmetric argument involving the right side of the box proves (ii).

    For the first assertion, we need to prove that no two rectangles in $\Q'$ overlap.
    The only case when this could a priori happen is if $R_1\in X$ and $R_2\in Y$ are overlapping after the shift.
    This would mean that $x(R_1)+w(R_1)< x(R_2)-\ell$. 
    However, again in $\Q$ there cannot be any rectangle lying in between $R_1$ and $R_2$, because every rectangle has width at least~$\ell$.
    Therefore, the $R_1$ and $R_2$ must see each other, which is a contradiction because $R_1\in X$ and~$R_2\in Y$.
\end{proof}
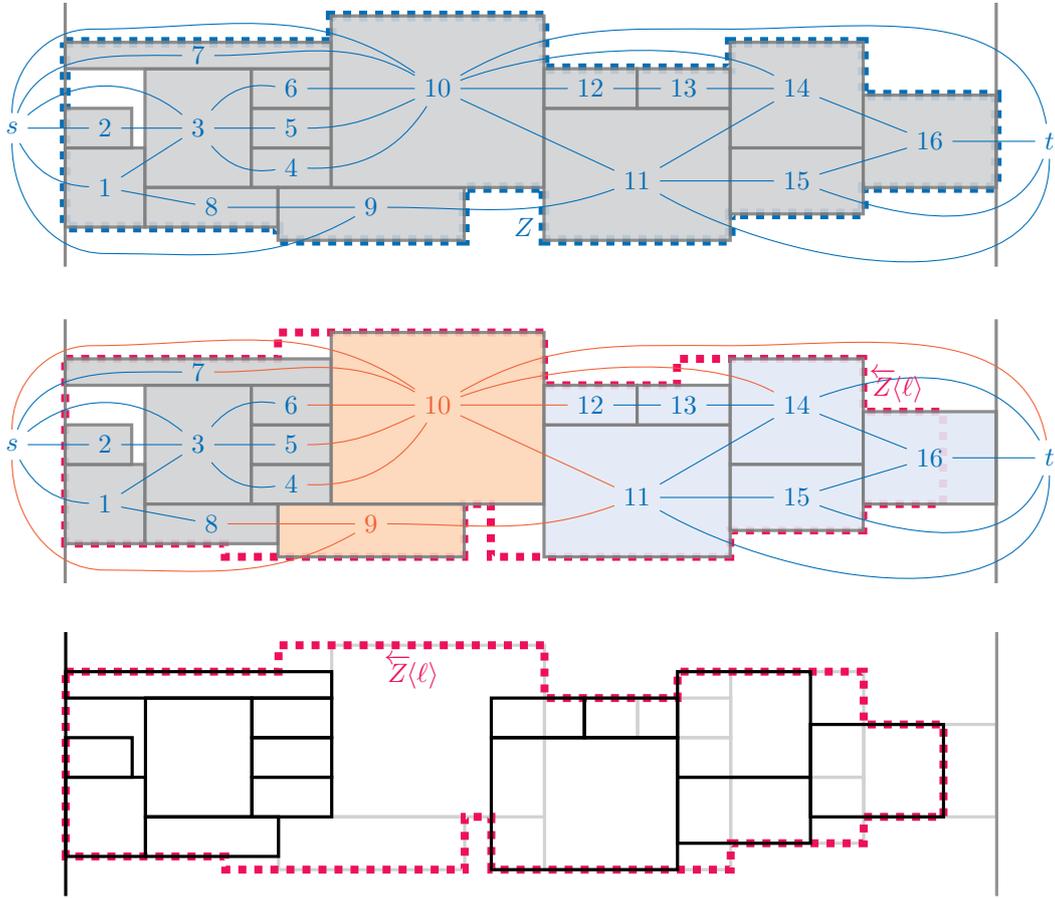
\begin{figure}[H]
    \centering
    \begin{tikzpicture}[scale=0.175,
        rect/.style={very thick,opacity=0.9,draw=black!50,fill=Gray!40},graph/.style={color=NavyBlue}]
        \draw[line width=4pt,dashed,color=NavyBlue] (0,17) -- (20,17) -- (20,19) -- (36,19) -- (36,15) -- (50,15) -- (50,17) -- (60,17) -- (60,13) -- (70,13) -- (70,6) -- (60,6) -- (60,4) -- (50,4) -- (50,2) -- (36,2) -- (36,6) -- (30,6) -- (30,2) -- (16,2) -- (16,3) -- (6,3) -- (0,3) -- cycle;

        \node at (34.5,3) {\color{NavyBlue}{$Z$}};
        
        \draw[style=rect] (0,0) -- (0,20);
        \draw[style=rect] (0,3) rectangle (6,9);
        \node[style=graph] (1) at (3,6) {1};
        \draw[style=rect] (0,9) rectangle (5,12);
        \node[style=graph] (2) at (3,10.5) {2};
        \draw[style=rect] (6,6) rectangle (14,15);
        \node[style=graph] (3) at (10,10.5) {3};
        \draw[style=rect] (14,6) rectangle (20,9);
        \node[style=graph] (4) at (17,7.5) {4};
        \draw[style=rect] (14,9) rectangle (20,12);
        \node[style=graph] (5) at (17,10.5) {5};
        \draw[style=rect] (14,12) rectangle (20,15);
        \node[style=graph] (6) at (17,13.5) {6};
        \draw[style=rect] (0,15) rectangle (20,17);
        \node[style=graph] (7) at (10,16) {7};
        \draw[style=rect] (6,6) rectangle (16,3);
        \node[style=graph] (8) at (11,4.5) {8};
        \draw[style=rect] (16,6) rectangle (30,2);
        \node[style=graph] (9) at (23,4.5) {9};
        \draw[style=rect] (20,6) rectangle (36,19);
        \node[style=graph] (10) at (28,13.5) {10};
        \draw[style=rect] (36,2) rectangle (50,12);
        \node[style=graph] (11) at (43,6.5) {11};
        \draw[style=rect] (36,12) rectangle (43,15);
        \node[style=graph] (12) at (39.5,13.5) {12};
        \draw[style=rect] (43,12) rectangle (50,15);
        \node[style=graph] (13) at (46.5,13.5) {13};
        \draw[style=rect] (50,9) rectangle (60,17);
        \node[style=graph] (14) at (55,13.5) {14};
        \draw[style=rect] (50,9) rectangle (60,4);
        \node[style=graph] (15) at (55,6.5) {15};
        \draw[style=rect] (60,6) rectangle (70,13);
        \node[style=graph] (16) at (65,9.5) {16};
        \draw[style=rect] (70,0) -- (70,20);
        \draw[style=graph] (11) to (14);
        \node[style=graph] (s) at (-4,10.5) {$s$};
        \node[style=graph] (t) at (74,9.5) {$t$};
        \draw[style=graph] (s) to[bend right] (1);
        \draw[style=graph] (s) to (2);
        \draw[style=graph] (s) to[out=70,in=180] (7);
        \draw[style=graph] (s) to[out=40,in=140] (3);
        \draw[style=graph] (s) to[out=270,in=180] (3,1) to[out=0,in=210] (9);
        \draw[style=graph] (s) to[out=90,in=180] (2,18) to[out=0,in=145] (10);
        \draw[style=graph] (2) to (3);
        \draw[style=graph] (1) to (3);
        \draw[style=graph] (1) to (8);
        \draw[style=graph] (3) to[bend left] (6);
        \draw[style=graph] (3) to (5);
        \draw[style=graph] (3) to[bend right] (4);
        \draw[style=graph] (8) to (9);
        \draw[style=graph] (7) to[out=0,in=150] (10);
        \draw[style=graph] (6) to (10);
        \draw[style=graph] (5) to[out=0,in=205] (10);
        \draw[style=graph] (4) to[bend right] (10);
        \draw[style=graph] (10) to (12);
        \draw[style=graph] (10) to (11);
        \draw[style=graph] (9) to[out=0,in=200] (11);
        \draw[style=graph] (10) to[bend left,looseness=0.7,out=20] (14);
        \draw[style=graph] (12) to (13);
        \draw[style=graph] (13) to (14);
        \draw[style=graph] (11) to (15);
        \draw[style=graph] (14) to (16);
        \draw[style=graph] (15) to (16);
        \draw[style=graph] (16) to (t);
        \draw[style=graph] (15) to[out=340,in=250] (t);
        \draw[style=graph] (11) to[out=330,in=270] (t);
        \draw[style=graph] (10) to[out=30,in=180] (50,18) to[out=0,in=100] (t);
    \end{tikzpicture}
    \begin{tikzpicture}[scale=0.175,
        rect/.style={very thick,opacity=0.9,draw=black!50,fill=Gray!40},graph/.style={color=NavyBlue},separator/.style={color=RedOrange}]
        \draw[line width=3pt,dashed,,color=OrangeRed] (0,17) -- (16,17) -- (16,19) -- (36,19) -- (36,15) -- (46,15) -- (46,17) -- (60,17) -- (60,13) -- (66,13) -- (66,6) -- (60,6) -- (60,4) -- (50,4) -- (50,2) -- (32,2) -- (32,6) -- (30,6) -- (30,2) -- (12,2) -- (12,3) -- (6,3) -- (0,3) -- cycle;

        \draw[style=rect] (0,0) -- (0,20);
        \draw[style=rect] (0,3) rectangle (6,9);
        \node[style=graph] (1) at (3,6) {1};
        \draw[style=rect] (0,9) rectangle (5,12);
        \node[style=graph] (2) at (3,10.5) {2};
        \draw[style=rect] (6,6) rectangle (14,15);
        \node[style=graph] (3) at (10,10.5) {3};
        \draw[style=rect] (14,6) rectangle (20,9);
        \node[style=graph] (4) at (17,7.5) {4};
        \draw[style=rect] (14,9) rectangle (20,12);
        \node[style=graph] (5) at (17,10.5) {5};
        \draw[style=rect] (14,12) rectangle (20,15);
        \node[style=graph] (6) at (17,13.5) {6};
        \draw[style=rect] (0,15) rectangle (20,17);
        \node[style=graph] (7) at (10,16) {7};
        \draw[style=rect] (6,6) rectangle (16,3);
        \node[style=graph] (8) at (11,4.5) {8};
        \draw[style=rect,fill=Orange!30] (16,6) rectangle (30,2);
        \node[style=separator] (9) at (23,4.5) {9};
        \draw[style=rect,fill=Orange!30] (20,6) rectangle (36,19);
        \node[style=separator] (10) at (28,13.5) {10};
        \draw[style=rect,fill=NavyBlue!10] (36,2) rectangle (50,12);
        \node[style=graph] (11) at (43,6.5) {11};
        \draw[style=rect,fill=NavyBlue!10] (36,12) rectangle (43,15);
        \node[style=graph] (12) at (39.5,13.5) {12};
        \draw[style=rect,fill=NavyBlue!10] (43,12) rectangle (50,15);
        \node[style=graph] (13) at (46.5,13.5) {13};
        \draw[style=rect,fill=NavyBlue!10] (50,9) rectangle (60,17);
        \node[style=graph] (14) at (55,13.5) {14};
        \draw[style=rect,fill=NavyBlue!10] (50,9) rectangle (60,4);
        \node[style=graph] (15) at (55,6.5) {15};
        \draw[style=rect,fill=NavyBlue!10] (60,6) rectangle (70,13);
        \node[style=graph] (16) at (65,9.5) {16};
        \draw[style=rect] (70,0) -- (70,20);
        \draw[style=graph] (11) to (14);
        \node[style=graph] (s) at (-4,10.5) {$s$};
        \node[style=graph] (t) at (74,9.5) {$t$};
        \draw[style=graph] (s) to[bend right] (1);
        \draw[style=graph] (s) to (2);
        \draw[style=graph] (s) to[out=70,in=180] (7);
        \draw[style=graph] (s) to[out=40,in=140] (3);
        \draw[style=separator] (s) to[out=270,in=180] (3,1) to[out=0,in=210] (9);
        \draw[style=separator] (s) to[out=90,in=180] (2,18) to[out=0,in=145] (10);
        \draw[style=graph] (2) to (3);
        \draw[style=graph] (1) to (3);
        \draw[style=graph] (1) to (8);
        \draw[style=graph] (3) to[bend left] (6);
        \draw[style=graph] (3) to (5);
        \draw[style=graph] (3) to[bend right] (4);
        \draw[style=separator] (8) to (9);
        \draw[style=separator] (7) to[out=0,in=150] (10);
        \draw[style=separator] (6) to (10);
        \draw[style=separator] (5) to[out=0,in=205] (10);
        \draw[style=separator] (4) to[bend right] (10);
        \draw[style=separator] (10) to (12);
        \draw[style=separator] (10) to (11);
        \draw[style=separator] (9) to[out=0,in=200] (11);
        \draw[style=separator] (10) to[bend left,looseness=0.7,out=20] (14);
        \draw[style=graph] (12) to (13);
        \draw[style=graph] (13) to (14);
        \draw[style=graph] (11) to (15);
        \draw[style=graph] (14) to (16);
        \draw[style=graph] (15) to (16);
        \draw[style=graph] (16) to (t);
        \draw[style=graph] (15) to[out=340,in=250] (t);
        \draw[style=graph] (14) to[out=20,in=120] (t);
        \draw[style=graph] (11) to[out=330,in=270] (t);
        \draw[style=separator] (10) to[out=30,in=180] (50,18) to[out=0,in=100] (t);

        \node at (62.5,15) {\color{OrangeRed}{$\overleftarrow{Z}\!\langle \ell \rangle$}};

    \end{tikzpicture}
    \begin{tikzpicture}[scale=0.175]        
        \draw[very thick,color=Gray!40] (16,6) rectangle (30,2);
        \draw[very thick,color=Gray!40] (20,6) rectangle (36,19);
        \draw[very thick,color=Gray!40] (36,2) rectangle (50,12);
        \draw[very thick,color=Gray!40] (36,12) rectangle (43,15);
        \draw[very thick,color=Gray!40] (43,12) rectangle (50,15);
        \draw[very thick,color=Gray!40] (50,9) rectangle (60,17);
        \draw[very thick,color=Gray!40] (50,9) rectangle (60,4);
        \draw[very thick,color=Gray!40] (60,6) rectangle (70,13);
        
        \draw[line width=3pt,dashed,,color=OrangeRed] (0,17) -- (16,17) -- (16,19) -- (36,19) -- (36,15) -- (46,15) -- (46,17) -- (60,17) -- (60,13) -- (66,13) -- (66,6) -- (60,6) -- (60,4) -- (50,4) -- (50,2) -- (32,2) -- (32,6) -- (30,6) -- (30,2) -- (12,2) -- (12,3) -- (6,3) -- (0,3) -- cycle;

        \node at (26,17) {\color{OrangeRed}{$\overleftarrow{Z}\!\langle \ell \rangle$}};
        
        \draw[very thick] (0,0) -- (0,20);
        \draw[very thick] (0,3) rectangle (6,9);
        \draw[very thick] (0,9) rectangle (5,12);
        \draw[very thick] (6,6) rectangle (14,15);
        \draw[very thick] (14,6) rectangle (20,9);
        \draw[very thick] (14,9) rectangle (20,12);
        \draw[very thick] (14,12) rectangle (20,15);
        \draw[very thick] (0,15) rectangle (20,17);
        \draw[very thick] (6,6) rectangle (16,3);
        \draw[very thick,color=Gray] (70,0) -- (70,20);
        
        \draw[very thick] (32,2) rectangle (46,12);
        \draw[very thick] (32,12) rectangle (39,15);
        \draw[very thick] (39,12) rectangle (46,15);
        \draw[very thick] (46,9) rectangle (56,17);
        \draw[very thick] (46,9) rectangle (56,4);
        \draw[very thick,] (56,6) rectangle (66,13);

    \end{tikzpicture}
    \caption{Illustration of the proof of \Cref{prop:z-}. From top to bottom:
    First, a blue conflict graph of the gray packing. The packing is entirely in $Z$, delimited by the blue border.
    Secondly, we remove the orange separator $C$ and want to pack the leftover rectangles inside the red region $\protect\overleftarrow{Z}\!\langle \ell \rangle$. The set of rectangles referred to as $Y$ in the proof is in light blue, and $X$ is left gray.
    Finally, we pack in $\protect\overleftarrow{Z}\!\langle \ell \rangle$ by shifting the rectangles at the right of the separator by $\ell$ to the left.}
    \label{fig:z-}
\end{figure}

Now that we have emptied a strip to the right of the zone, we can do some resource augmentation in order to replace the original rectangles by their rounded versions, while still being able to pack them inside the rounded zone.

\begin{proposition}\label{prop:z+}
Let $\Q$ be a packing in a zone $Z\subseteq [0,N_1-\ell]\times [0,N_2]$ such that every rectangle of $\Q$ has width at least $2\ell$. 
    Then $\roundell{\Q}$ can be packed inside the zone $\overrightarrow{Z}\!\langle \ell \rangle$.
\end{proposition}
For clarification, see \Cref{fig:z+}. 
\begin{proof}
    First, scale horizontally every rectangle in $\Q$ by a factor $\lambda = 1+\ell/N_1$, i.e., for a rectangle $R=[x,x+w]\times [y,y+h]$ we define the rectangle $R^\times = [\lambda x,\lambda(x+w)]\times [y,y+h]$.
    These rectangles fit inside $\roundell{Z}$. Indeed, the maximum possible displacement of a point is $N_1\cdot \ell/N_1 =  \ell$, i.e. the image of a point under scaling is at horizontal distance at most $\ell$ to the right of the original point. Next, observe that every rectangle $\roundell{R}$ can be entirely placed inside the corresponding rectangle $R^\times$, because $$\lambda w = w + w \ell/N_1 \geq w + \ell^2/N_1 = \ell'+ w = \ell'(1+w/\ell') \geq \ell'\ceil{w/\ell'}.$$ (Recall here that we assumed all rectangles to have width at least $\ell$.)
    Now, $\Q'$ can be obtained from $\Q$ by replacing each $R\in \Q'$ with $R^{\times}$, fitting $\roundell{R}$ inside $R^{\times}$, and finally shifting all rectangles to the left so that they have integer coordinates.
    The last step is always possible as every rectangle has integer length.
\end{proof}
\begin{figure}[H]
    \centering
    \begin{tikzpicture}[scale=0.175,
            rect/.style={very thick,draw=black!50,fill=Gray!40}
        ]
        \pgfdeclarepatternformonly{swnestripes}{\pgfpoint{0cm}{0cm}}{\pgfpoint{1cm}{1cm}}{\pgfpoint{1cm}{1cm}}
        {
            \foreach \i in {0.1cm, 0.3cm,...,0.9cm}
            {
             \pgfpathmoveto{\pgfpoint{\i}{0cm}}
             \pgfpathlineto{\pgfpoint{1cm}{1cm - \i}}
             \pgfpathlineto{\pgfpoint{1cm}{1cm - \i + 0.1cm}}
             \pgfpathlineto{\pgfpoint{\i - 0.1cm}{0cm}}
             \pgfpathclose%
             \pgfusepath{fill}
             \pgfpathmoveto{\pgfpoint{0cm}{\i}}
             \pgfpathlineto{\pgfpoint{1cm - \i}{1cm}}
             \pgfpathlineto{\pgfpoint{1cm - \i - 0.1cm}{1cm}}
             \pgfpathlineto{\pgfpoint{0cm}{\i + 0.1cm}}
             \pgfpathclose%
             \pgfusepath{fill}
            }
        }
        \pattern[pattern=swnestripes, pattern color=RedOrange!50] (66,0) rectangle (70,20);

        \draw[line width=3pt,dashed,color=NavyBlue] (0,17) -- (20,17) -- (20,15) -- (36,15) -- (46,15) -- (46,17) -- (56,17) -- (56,13) -- (66,13) -- (66,6) -- (56,6) -- (56,4) -- (46,4) -- (46,2) -- (32,2) -- (32,3) -- (0,3) -- cycle;

        \node at (26,9) {\color{NavyBlue} $Z$};
        
        \draw[style=rect] (0,0) -- (0,20);
        \draw[style=rect] (0,3) rectangle (6,9);
        \draw[style=rect] (0,9) rectangle (5,12);
        \draw[style=rect] (6,6) rectangle (14,15);
        \draw[style=rect] (14,6) rectangle (20,9);
        \draw[style=rect] (14,9) rectangle (20,12);
        \draw[style=rect] (14,12) rectangle (20,15);
        \draw[style=rect] (0,15) rectangle (20,17);
        \draw[style=rect] (6,6) rectangle (16,3);
        \draw[style=rect] (70,0) -- (70,20);
        
        \draw[style=rect] (32,2) rectangle (46,12) node[pos=0.5] {$R$};
        \draw[style=rect] (32,12) rectangle (39,15);
        \draw[style=rect] (39,12) rectangle (46,15);
        \draw[style=rect] (46,9) rectangle (56,17);
        \draw[style=rect] (46,9) rectangle (56,4);
        \draw[style=rect] (56,6) rectangle (66,13);
    \end{tikzpicture}
    \begin{tikzpicture}[scale=0.175,
            rect/.style={very thick,draw=black!50,fill=Gray!40}
        ]
        \draw[line width=3pt,dashed,color=OrangeRed] (0,17) -- (24,17) -- (24,15) -- (36,15) -- (46,15) -- (46,17) -- (60,17) -- (60,13) -- (70,13) -- (70,6) -- (60,6) -- (60,4) -- (50,4) -- (50,2) -- (32,2) -- (32,3) -- (0,3) -- cycle;

        \node at (28,9) {\color{OrangeRed} $\overrightarrow{Z}\!\langle \ell \rangle$};

        \newcommand{\factor}{1.05714285714}; % lambda = 1 + \ell/N_2; \ell=4; N_2=70
        
        \draw[style=rect] (0,0) -- (0,20);
        \draw[style=rect] (0*\factor,3) rectangle (6*\factor,9);
        \draw[style=rect] (0*\factor,9) rectangle (5*\factor,12);
        \draw[style=rect] (6*\factor,6) rectangle (14*\factor,15);
        \draw[style=rect] (14*\factor,6) rectangle (20*\factor,9);
        \draw[style=rect] (14*\factor,9) rectangle (20*\factor,12);
        \draw[style=rect] (14*\factor,12) rectangle (20*\factor,15);
        \draw[style=rect] (0*\factor,15) rectangle (20*\factor,17);
        \draw[style=rect] (6*\factor,6) rectangle (16*\factor,3);
        \draw[style=rect] (70,0) -- (70,20);
        
        \draw[style=rect] (32*\factor,2) rectangle (46*\factor,12) node[pos=0.5] {$R^\times$};
        \draw[style=rect] (32*\factor,12) rectangle (39*\factor,15);
        \draw[style=rect] (39*\factor,12) rectangle (46*\factor,15);
        \draw[style=rect] (46*\factor,9) rectangle (56*\factor,17);
        \draw[style=rect] (46*\factor,9) rectangle (56*\factor,4);
        \draw[style=rect] (56*\factor,6) rectangle (66*\factor,13);

        \draw [thick, -{Latex[round]}] (62,3) -- (68,3) node[midway,below] {\footnotesize$\times (1+\ell/N_1)$};
    \end{tikzpicture}
    \begin{tikzpicture}[scale=0.175,
            rect/.style={very thick,draw=black!50,fill=Gray!40}
        ]
        \draw[line width=3pt,dashed,color=OrangeRed] (0,17) -- (24,17) -- (24,15) -- (36,15) -- (46,15) -- (46,17) -- (60,17) -- (60,13) -- (70,13) -- (70,6) -- (60,6) -- (60,4) -- (50,4) -- (50,2) -- (32,2) -- (32,3) -- (0,3) -- cycle;

        \newcommand{\factor}{1.05714285714}; % lambda = 1 + \ell/N_2; \ell=4; N_2=70
        \newcommand{\reducedfactor}{0.05714285714}; % lambda = 1 + \ell/N_2; \ell=4; N_2=70

        \node at (28,9) {\color{OrangeRed} $\overrightarrow{Z}\!\langle \ell \rangle$};
        
        \draw[style=rect] (0,0) -- (0,20);
        \draw[style=rect] (0*\factor,3) rectangle (6+0*\reducedfactor,9);
        \draw[style=rect] (0*\factor,9) rectangle (5+0*\reducedfactor,12);
        \draw[style=rect] (6*\factor,6) rectangle (14+6*\reducedfactor,15);
        \draw[style=rect] (14*\factor,6) rectangle (20+14*\reducedfactor,9);
        \draw[style=rect] (14*\factor,9) rectangle (20+14*\reducedfactor,12);
        \draw[style=rect] (14*\factor,12) rectangle (20+14*\reducedfactor,15);
        \draw[style=rect] (0*\factor,15) rectangle (20+0*\reducedfactor,17);
        \draw[style=rect] (6*\factor,6) rectangle (16+6*\reducedfactor,3);
        \draw[style=rect] (70,0) -- (70,20);
        
        \draw[style=rect] (32*\factor,2) rectangle (46+32*\reducedfactor,12) node[pos=0.5] {$\roundell{R}$};
        \draw[style=rect] (32*\factor,12) rectangle (39+32*\reducedfactor,15);
        \draw[style=rect] (39*\factor,12) rectangle (46+39*\reducedfactor,15);
        \draw[style=rect] (46*\factor,9) rectangle (56+46*\reducedfactor,17);
        \draw[style=rect] (46*\factor,9) rectangle (56+46*\reducedfactor,4);
        \draw[style=rect] (56*\factor,6) rectangle (66+56*\reducedfactor,13);
    \end{tikzpicture}
    \caption{Illustration of the proof of Proposition \ref{prop:z+}. From top to bottom: First, $\Q$ is packed into $Z$ (blue zone). Then, the rectangles in $\Q$ are scaled horizontally by a factor $\lambda=1+\ell/N_1$. We argue in the proof that these scaled-up rectangles are packed in $\protect\overrightarrow{Z}\!\langle \ell \rangle$ (red zone). Finally, we replace each scaled-up rectangle by its rounded version, which has smaller width.  }
    \label{fig:z+}
\end{figure}
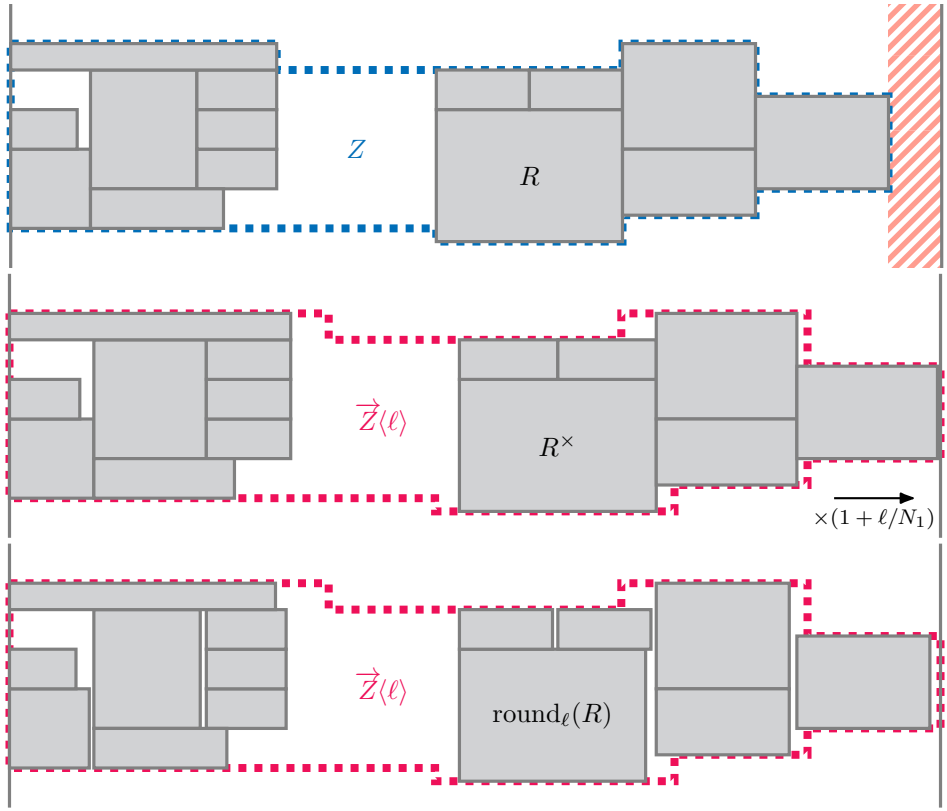

We may now combine \Cref{prop:z-} and \Cref{prop:z+} to achieve our goal.

\begin{proof}[Proof of \Cref{lem:approx_widths}]
    Apply  \Cref{prop:z-} and \Cref{prop:z+} to get that $\roundell{\Q\setminus C}$ can be packed in $\overrightarrow{Z'}\!\langle \ell \rangle$, where $Z'= \overleftarrow{Z}\!\langle \ell \rangle$.
    As $\overrightarrow{Z'}\!\langle \ell \rangle = \overleftrightarrow{Z}\!\langle \ell \rangle$, the proof is finished.
\end{proof}
\subsection{Proof of the \texorpdfstring{\hyperref[lem:structural]{Structural Lemma}}{structural lemma}}

Finally, in this subsection we define the polylines that we are interested in and prove some results about zones and polylines to finish the proof of \Cref{lem:structural}.
The main idea is to construct some well-chosen polylines by looking at the rectangles on short $s$-$t$ paths.
These polylines are then used to delimit zones in which we can find a separator of bounded size, and apply the ideas of the previous subsections.

Recall that we are working with a packing $\S$ of size $k$ consisting of rectangles of width at least $2\ell$ each. Let $G$ be the conflict graph of $\S$. For every $R\in \S$, by $v_R$ we denote the vertex of $G$ corresponding to $R$.
First, we need to understand how $s$-$t$ paths in $G$ can be mapped to polylines.

\begin{definition}[Bottom polyline of a path]
    Consider an $s$-$t$ path $P = (s, v_{R_1}, v_{R_2}, \dots, v_{R_m}, t)$ in $G$, and suppose that for each $i\in \{0,1,\ldots,m\}$, that $R_i$ and $R_{i+1}$ see each other is witnessed by the segment $s_i=[x(R_i)+w(R_i),x(R_{i+1})] \times \{y_i\}$ (where $R_0=s$ and $R_{m+1}=t$). Then define the \emph{bottom polyline} of $P$ as the polyline $\P$ formed by the union of the following segments:
    \begin{itemize}[nosep]
        \item $[x(R_i), x(R_i)+w(R_i)]\times\{y(R_i)\}$ for each $i\in [m]$,
        \item $\{x(R_i)\} \times [\min\{y(R_i),y_{i-1}\},\max\{y(R_i),y_{i-1}\}]$ for each $i\in [m]$,
        \item $\{x(R_i)+w(R_i)\} \times [\min\{y(R_{i}),y_i\},\max\{y(R_{i}),y_i\}]$ for each $i\in [m]$, and
        \item $s_i$ for each $i\in \{0,1,\ldots,m\}$.
    \end{itemize}
    Less formally, $\P$ is the union of the segments $s_i$ joining the rectangles of the path, the bottom sides of the rectangles, and parts of the left/right sides of the rectangles to join the segments to the bottom sides. 
\end{definition}
Similarly, we define the notion of the top polyline of an $s$-$t$ path in $G$. When defining at the same time the top and the bottom polyline of the same path, we always use the same segments $s_i$ to define how rectangles $R_i$ and $R_{i+1}$ should be linked.
Finally, we will also need the middle polyline.

\begin{definition}[Middle polyline of a path]
    Consider an $s$-$t$ path $P = (s, v_{R_1}, v_{R_2}, \dots, v_{R_m}, t)$ in $G$, and suppose that for each $i\in \{0,1,\ldots,m\}$, that $R_i$ and $R_{i+1}$ see each other is witnessed by the segment $s_i=[x(R_i)+w(R_i),x(R_{i+1})] \times \{y_i\}$ (where $R_0=s$ and $R_{m+1}=t$). Then define the \emph{middle polyline} of $P$ as the polyline $\P$ formed by the union of the following segments:
    \begin{itemize}[nosep]
        \item $\{x(R_i)+w(R_i)/2\} \times [\min\{y(R_i)+h(R_i)/2,y_i\},\max\{y(R_i)+h(R_i)/2,y_i\}]$ for each $i\in [m]$,
        \item $\{x(R_i)+w(R_i)/2\} \times [\min\{y(R_i)+h(R_i)/2,y_{i+1}\},\max\{y(R_i)+h(R_i)/2,y_{i+1}\}]$ for each $i\in [m]$, and
        \item $[\max(x(R_i)+w(R_i)/2,0),\min(x(R_{i+1})+w(R_{i+1})/2,N_2)] \times \{y_i\}$ for each $i\in \{0,1,\ldots,m\}$.
    \end{itemize}
    Less formally, $\P$ is the union of:
    \begin{itemize}[nosep]
        \item a vertical segment from the center of each $R_i$ to the vertical position of $s_i$,
        \item a vertical segment from the center of each $R_i$ to the vertical position of $s_{i+1}$,
        \item all segments $s_i$ extended so that they reach the horizontal coordinates of the centers of the corresponding rectangles.
    \end{itemize}
\end{definition}
For a visual representation, see \Cref{fig:polylines}. 
\begin{figure}[H]
    \centering

    \begin{tikzpicture}[scale=0.2]
        \draw[line width=2pt,color=Orange,draw opacity=0.7] (0,10.2) -- (4.8,10.2) -- (4.8,14.2) -- (11.2,14.2) -- (11.2,10.2) -- (19.8,10.2) -- (19.8,18.2) -- (30.2,18.2) -- (30.2,12.2) -- (37.8,12.2) -- (37.8,16.2) -- (46.2,16.2) -- (46.2,14.2) -- (54.2,14.2) -- (54.2,10.2) -- (60,10.2) ;

        \draw[line width=2pt,color=NavyBlue,draw opacity=0.7] (0,9.8) -- (4.8,9.8) -- (4.8,5.8) -- (11.2,5.8) -- (11.2,9.8) -- (19.8,9.8) -- (19.8,5.8) -- (29.8,5.8) -- (29.8,1.8) -- (38.2,1.8) -- (38.2,8.8) -- (45.8,8.8) -- (45.8,5.8) -- (54.2,5.8) -- (54.2,9.8) -- (60,9.8) ;

        \draw[very thick] (0,0) -- (0,20);
        \draw[very thick] (0,10) -- (5,10);
        
        \node at (8,15.5) {\color{Orange}{\textbf{top}}};
        \draw[very thick] (5,6) rectangle (11,14);
        \node at (6.5,7.5) {$R_1$};
        \draw[dashed] (8,6) -- (8,14);
        \draw[dashed] (5,10) -- (11,10);
        
        \node at (15.5,11.5) {\color{ForestGreen}{\textbf{middle}}};
        \draw[very thick] (11,10) -- (20,10);

        \node at (25,4.5) {\color{NavyBlue}{\textbf{bottom}}};
        \draw[very thick] (20,6) rectangle (30,18);
        \node at (23.5,14) {$R_2$};
        \draw[dashed] (25,6) -- (25,18);
        \draw[dashed] (20,12) -- (30,12);
        
        \draw[very thick] (30,2) rectangle (38,12);
        \node at (32.5,5) {$R_3$};
        \draw[dashed] (34,2) -- (34,12);
        \draw[dashed] (30,7) -- (38,7);
        
        \draw[very thick] (38,9) rectangle (46,16);
        \node at (40.5,14) {$R_4$};
        \draw[dashed] (42,9) -- (42,16);
        \draw[dashed] (38,12.5) -- (46,12.5);
        
        \draw[very thick] (46,6) rectangle (54,14);
        \node at (48.5,8) {$R_5$};
        \draw[dashed] (50,6) -- (50,14);
        \draw[dashed] (46,10) -- (54,10);
        
        \draw[very thick] (54,10) -- (60,10);
        \draw[very thick] (60,0) -- (60,20);

        \draw[line width=2pt,color=ForestGreen,draw opacity=0.7] (0,10) -- (25,10) -- (25,12) -- (25,7) -- (34,7) -- (34,10) -- (42,10) -- (42,12.5) -- (42, 11) -- (50,11) -- (50,10) -- (60,10) ;

    \end{tikzpicture}
    \caption{In orange, the top polyline of the $st$-path formed by $R_1,R_2,R_3,R_4$ and $R_5$. In blue, its bottom polyline, and in green, its middle polyline. The dashed lines split their respective rectangles into 4 equal parts.}
    \label{fig:polylines}
\end{figure}
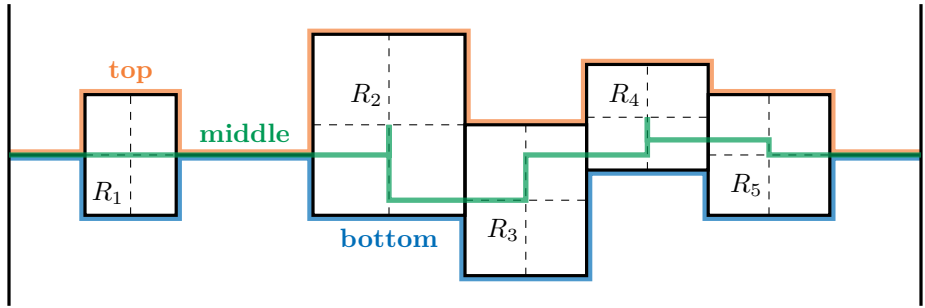

The following is clear.
\begin{proposition}\label{prop:polyline_complexity}
    The top, bottom and middle polylines of a path $P$ have complexity at most $4|P|+1$, where $|P|$ denotes the number of vertices on $P$.
\end{proposition}

Moreover, the middle polyline is defined so that we have space to the left and the right when performing resource augmentation. 
This will be made clear in the following definitions and lemmas; see 
\Cref{fig:crossing_polyline}.
\begin{figure}[H]
    \centering

    \begin{tikzpicture}[scale=0.2,
            rect/.style={very thick,draw=black!50,fill=Gray!40},
            old_orange_rect/.style={very thick,draw=Orange!20},
            orange_rect/.style={very thick,draw=Orange},
            line/.style={very thick,draw=black!50},
            graph/.style={circle,fill=NavyBlue,inner sep=2pt},
            edge/.style={draw=NavyBlue}
        ]
    
        \node (s) at (-2,12) {\color{NavyBlue}$s$};
    
        \draw[style=line] (0,-7) -- (0,20);
        \draw[style=rect] (0,10) -- (3,10);
        
        \draw[style=rect] (3,6) rectangle (11,14);
        \node[style=graph] (1) at (7,12) {};
        
        \draw[style=old_orange_rect] (11.2,9.8) rectangle (20,-5);
        \draw[style=orange_rect] (7.2,9.8) rectangle (16,-5);
        \node at (5,2) {\color{Orange}{$\overleftarrow{\Q}\!\langle \ell \rangle$}};

        \draw[style=old_orange_rect] (20,1.8) rectangle (38.2,-5);
        \draw[style=orange_rect] (16,1.8) rectangle (34.2,-5);
        \node at (29.1,-1.6) {\color{Orange!40}{$\Q$}};
        
        \draw[style=rect] (11,10) -- (20,10);
        \node at (18,11.5) {\color{ForestGreen}{$\P$}};
        
        \draw[style=rect] (20,6) rectangle (30,18);
        \node[style=graph] (2) at (25,15) {};
        
        \draw[style=rect] (30,2) rectangle (38,12);
        \node[style=graph] (3) at (34,11) {};
        
        \draw[style=rect] (38,9) rectangle (46,16);
        \node[style=graph] (4) at (42,14.25) {};
        
        \draw[style=rect] (46,6) rectangle (54,14);
        \node[style=graph] (5) at (50,12.5) {};
        
        \draw[style=rect] (54,10) -- (60,10);
        \draw[style=line] (60,-7) -- (60,20);

        \node (t) at (62,12.5) {\color{NavyBlue}$t$};
        
        \draw[style=old_orange_rect] (38.2,8.8) rectangle (45.8,-5);
        \draw[style=orange_rect] (34.2,8.8) rectangle (42,-5);
        
        \draw[style=old_orange_rect] (45.8,5.8) rectangle (59.8,-3);
        \draw[style=orange_rect] (42,5.8) rectangle (56,-3);
        
        \draw[style=old_orange_rect] (54.2,9) rectangle (59.8,6);
        \draw[style=orange_rect] (50.2,9) rectangle (56,6);

        \draw[style=edge] (s) to[out=0,in=180] (1);
        \draw[style=edge] (1) to[out=0,in=180] (2);
        \draw[style=edge] (2) to[out=0,in=180] (3);
        \draw[style=edge] (3) to[out=0,in=180] (4);
        \draw[style=edge] (4) to[out=0,in=180] (5);
        \draw[style=edge] (5) to[out=0,in=180] (t);

        \node at (15.5,15) {\color{NavyBlue}{$P$}};

        \draw[line width=2pt,color=ForestGreen,draw opacity=0.7] (0,10) -- (25,10) -- (25,12) -- (25,7) -- (34,7) -- (34,10) -- (42,10) -- (42,12.5) -- (42, 11) -- (50,11) -- (50,10) -- (60,10) ;

    \end{tikzpicture}
    \caption{In green, the middle polyline $\P$ of the blue $st$-path $P$ constituted of the gray rectangles. In light orange, $\Q$, and in orange, $\protect\overleftarrow{\Q}\!\langle \ell \rangle$. Notice that the green polyline does not cross any orange rectangles.}
    \label{fig:crossing_polyline}
\end{figure}
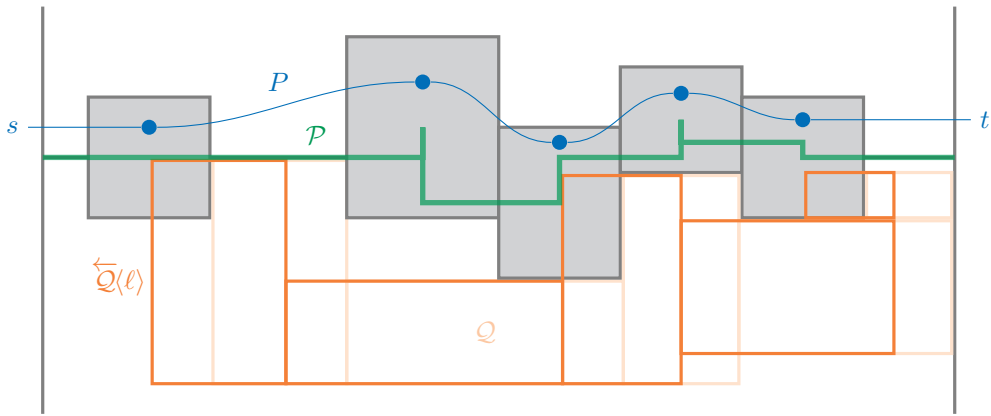

\begin{proposition}\label{prop:middle_poly_crossing}
    Suppose $P$ is an $s$-$t$ path in the conflict graph $G$ of the packing $\S$. Let $\P$ be the middle polyline of $P$ and let $\Q$ be the packing obtained from $\S$ by removing all the rectangles participating in $P$.
    Then $\P$ does not cross $\overleftarrow{\Q}\!\langle \ell \rangle$.
    The same goes for $\overrightarrow{\Q}\!\langle \ell \rangle$, and therefore also for $\overleftrightarrow{\Q}\!\langle \ell \rangle$.
\end{proposition}
\begin{proof}
    Suppose $\P$ crosses $\overleftarrow{R}\!\langle \ell \rangle$ for some $R\in\Q$.
    Then there exists $R_1,R_2\in V(P)$ (which are possibly the left or the right side of the box) such that $R_1$ and $R_2$ see each other through a segment $s = [x(R_1)+w(R_1),x(R_2)] \times \{y\}$ and $R$ crosses one of the following segments:
    \begin{enumerate}
        \item $\{x(R_1)+w(R_1)/2\} \times [\min\{y(R_1)+h(R_1)/2,y\},\max\{y(R_1)+h(R_1)/2,y\}]$,
        \item $[x(R_1)+w(R_1)/2,x(R_2)+w(R_2)/2] \times \{y\}$,
        \item $\{x(R_2)+w(R_2)/2\} \times [\min\{y,y(R_2)+h(R_2)/2\},\max\{y,y(R_2)+h(R_2)/2\}]$.
    \end{enumerate}
    We show that every case leads to a contradiction.
    \begin{enumerate}
        \item Assume case 1. $\P$ crosses $\overleftarrow{R}\!\langle \ell \rangle$ but not $R$ so $x(R) \geq x(R_1)+w(R_1)/2$ and $x(R)-\ell \leq x(R_1)+w(R_1)/2$. Therefore $x(R_1) \leq x(R) \leq x(R_1)+w(R_1)/2+\ell \leq x(R_1)+w(R_1)$ because $w(R_1)\geq 2\ell$. Moreover, $[\min\{y(R_1)+h(R_1)/2,y\},\max\{y(R_1)+h(R_1)/2,y\}]\subseteq[y(R_1),y(R_1)+h(R_1)]$ by the definition of $y$. This means that $R$ and $R_1$ intersect at $(x(R),y')$ where $y'\in[y(R),y(R)+h(R)]\cap [\min\{y(R_1)+h(R_1)/2,y\},\max\{y(R_1)+h(R_1)/2,y\}]$, which is not possible.
        \item Assume case 2. This would mean that $y\in [y(R),y(R)+h(R)]$, $x(R) \geq x(R_2)+w(R_2)/2$ and $x(R)-\ell \leq x(R_2)+w(R_2)/2$ because $|[x(R_1)+w(R_1)/2,x(R_2)+w(R_2)/2]|\geq 2\ell$ and $\P$ crosses $\overleftarrow{R}\!\langle \ell \rangle$ but not $R$. Therefore $x(R_2) \leq x(R) \leq x(R_2)+w(R_2)/2+\ell \leq x(R_2)+w(R_2)$ because $w(R_2)\geq 2\ell$. This means that $R$ and $R_2$ intersect at $(x(R),y)$, which is not possible.
        \item Assume case 3. This is a similar argument as case 1, replacing $R_1$ by $R_2$.\qedhere
    \end{enumerate}
\end{proof}

Next, we need the following graph-theoretic observation.

\begin{proposition}\label{prop:separator}
    Let $G$ be a graph containing vertices $s$ and $t$. Suppose every $s$-$t$ path in $G$ at least $1/\eps$ internal vertices. 
    Then $G$ contains an $st$-separator of size at most $\eps (|V(G)|-2)$.
\end{proposition}
\begin{proof}
    As every $s$-$t$ path in $G$ contains at least $1/\eps$ internal vertices, one cannot find more than $\eps (|V(G)|-2)$ internally disjoint $s$-$t$ paths in $G$.
    By Menger's theorem, there is an $st$-separator of size at most $\eps (|V(G)|-2)$.
\end{proof}

We can now wrap up the section by proving the Structural Lemma.

\begin{proof}[Proof of \Cref{lem:structural}]
    Based on the assumed packing $\S$, we construct another packing $\S'$ and then we prove that it is structured and has size at least $(1-3\eps)k$. Let $G$ be the conflict graph of $\S$.
    Let $\F$ be an inclusion-wise maximal family $\F$ of internally disjoint $s$-$t$ paths in $G$, each with at most $1/\eps$ internal vertices. As the paths from $\F$ are internally disjoint, we can naturally enumerate them from bottom to top: $\F=\{P_1,P_2,\dots,P_m\}$. For convenience, let $P_0=P_{m+1}=\emptyset$.

    Because the conflict graph is planar by \Cref{lem:planarity}, by the Jordan Curve theorem, for each $i\in \{0,1,\ldots,m\}$ there is a set $V_i$ of vertices of $G$ that lies inside the cycle $P_i\cup P_{i+1}$.
    By construction of the conflict graph, $V_i$ is exactly the set of rectangles lying in the area $Z_i$ delimited by the box, the top polyline of $P_i$ and the bottom polyline of $P_{i+1}$.
    For each $V_i$ we construct a separating polyline $\P_i$ as follows:
    \begin{itemize}
        \item If $|V_{i-1}| \leq 1/\eps^2$ and $|V_i| \leq 1/\eps^2$, select the bottom polyline of $P_i$ as the separating~polyline.
        \item Otherwise, select the middle polyline of $P_i$ as the separating polyline.
    \end{itemize}
    All the polylines created are of complexity at most $4/\eps+1$ by \Cref{prop:polyline_complexity}.
    They partition the box into regions $B_0, B_1, \dots, B_{m+1}\subseteq B$, from the bottom to the top. Note that $B_i\supseteq Z_i$ for each relevant~$i$.

    Notice that in $G[V_i\cup \{s,t\}]$ there is no $s$-$t$ path of length at most $1/\eps$, because $\F$ is maximal.
    Let $C_i$ be the separator given by \Cref{prop:separator} for the graph $G[V_i\cup \{s,t\}]$. Then we have $|C_i|\leq \eps|V_i|$.
    
    We can now specify which rectangles we want include in $\S'$. We define $\S'$ to be the union of sets $V'_i$ for $i\in\{0,1,\ldots,m\}$, where
    \[V'_i=\begin{cases}
        V_i \cup V(P_i)\setminus \{s,t\} & \text{if } |V_{i-1}| \leq 1/\eps^2 \text{ and }|V_i| \leq 1/\eps^2, \\
        V_i & \text{if } |V_{i-1}| > 1/\eps^2 \text{ and }|V_i| \leq 1/\eps^2, \\
        V_i\setminus C_i & \text{otherwise.}
    \end{cases}\]
    We now argue that for each $i\in \{0,1,\ldots,m\}$, either $|V'_i|\leq 2/\eps^2$ and $V'_i\subseteq \S[B_i]$, or $\roundell{V'_i}$ can be packed in $B_i$. 

    First, observe that if 
    $|V_i|\leq 1/\eps^2$, then  $|V'_i|\leq |V_i\cup V(P_i)\setminus \{s,t\}| \leq 1/\eps^2+ 1/\eps \leq 2/\eps^2$. Further, if $|V_{i-1}|>1/\eps^2$ then $V_i'=V_i$ and trivially $\S[B_i]\supseteq \S[Z_i]=V_i$, and if $|V_{i-1}|\leq 1/\eps^2$ then $\P_i$ is the bottom polyline of $P_i$ and we have $\S[B_i]\supseteq V_i\cup V(P_i)\setminus \{s,t\}=V_i'$ as well. 

    Second, consider the case when $|V_i|> 1/\eps^2$. Notice that then $\P_i$ is the middle polyline of $P_i$ and $\P_{i+1}$ is the middle polyline of $P_{i+1}$, and $V'_i=V_i\setminus C_i$. Because $V_i$ can be packed inside $Z_i$, we can use \Cref{lem:approx_widths} on $V_i$ and $Z_i$ to pack $\roundell{V'_i}$ into $\overleftrightarrow{(Z_i)}\!\langle \ell \rangle$.
    By \Cref{prop:middle_poly_crossing}, we know that $\overleftrightarrow{(Z_i)}\!\langle \ell \rangle \subseteq B_i$, hence we can pack $\roundell{V'_i}$ into $B_i$.

    We conclude that indeed, $\S'$ is an $(\eps,\ell)$-structured packing, as witnessed by the polylines $\P_i$ for $i\in [m]$.
    What is left to show is that $|\S'|\geq (1-3\eps)k$. Call an index $i\in [m]$ {\em{heavy}} if $|V_i|>1/\eps^2$. Observe that
    $\S'\supseteq \S\setminus \bigcup_{i\colon \text{heavy}} C_i\cup V(P_i)\cup V(P_{i+1})$,
    hence it suffices to prove that
    $\left|\bigcup_{i\colon \text{heavy}} C_i\cup V(P_i)\cup V(P_{i+1})\setminus \{s,t\}\right|\leq 3\eps k$.
    Fix a heavy index $i$. First, observe that
    $|V(P_i)\cup V(P_{i+1})\setminus \{s,t\}|\leq 2/\eps\leq 2\eps|V_i|$, as each path $P_i$ has at most $1/\eps$ internal vertices.
    Second, by construction we have
    $|C_i|\leq \eps|V_i|$.
    Summing those inequalities throughout all heavy $i$ yields that 
    \[\left|\bigcup_{i\colon \text{heavy}} C_i\cup V(P_i)\cup V(P_{i+1})\setminus \{s,t\}\right|\leq \sum_{i\colon \text{heavy}} 3\eps |V_i|\leq 3\eps k,\]
    as required.
\end{proof}

\section{The algorithm}

In this section we finalize the proof of \Cref{thm:main}.
The section is divided into two parts. 
The first subsection describes an algorithm working under the assumption that the input set $\R$ only contains rectangles of width at least $2\ell$, for some $\ell > 0$.
In the second subsection, we show how to obtain the assumption that $\R$ only contains  rectangles of width at least $2\ell$ for $\ell = N_1/(\delta(B)k^2)$, at the expense of deleting an $\eps$ fraction of the rectangles in the packing. Therefore, we get a full algorithm as a corollary.

Throughout this section, fix an instance $(B,\R,k)$ of {\sc{2D Knapsack}}, where $B=[0,N_1]\times [0,N_2]$ and $\R$ consists of wide items.

\subsection{The algorithm for rectangles of substantial width}

The dynamic programming algorithm will gradually guess a good partition of the box into regions (that we know exists by \Cref{lem:structural}), and then solve the problem in each region independently. In order to avoid repeating the use of the same rectangles in different regions, we use color-coding.
\begin{definition}[Good coloring]
    Given a set of rectangles $\R$ and a subset $\S\subseteq \R$ of size $k$, a function $\col\colon \R\to[k]$ is a \emph{good coloring} for $\S$ if rectangles of $\S$ have pairwise different colors under $\col$.
\end{definition}
We cannot directly guess a good coloring of the rectangles, as a priori there are too many candidates.
We instead use the following classic result of Naor et al. \cite{naor1995splitters}, which says that there is only an fpt-sized family of candidates for a good coloring.

\begin{proposition}[Naor et al. \cite{naor1995splitters}]\label{prop:perfect_hash}
    For every set $\R$ and positive integer $k$, there exists a family $\F$ of colorings of $\R$ with color set $[k]$ such that $|\F|\leq e^k k^{\O(\log k)} \log |\R|$ and for every $\S\subseteq \R$ of size $k$, in $\F$ there is a good coloring for $\S$. Moreover, $\F$ can be computed in time $e^k k^{\O(\log k)} |\R|\log |\R|$.  
\end{proposition}

Next, we observe that once the number of different widths in the instance has been bounded, one can restrict attention to a small set of candidate rectangles. For this, notice the following: if we have a colored packing of size $k$ that contains a rectangle $R$, and in the packing we did not use another rectangle $R'$ of the same color and width as $R$, but satisfying $h(R')\leq h(R)$, then we can replace $R$ with $R'$ and we will still have a colored packing.
This observation leads to defining the following operation.
\begin{definition}[$\reduce_k(\R,\col)$]
    Suppose $\col\colon \R\to [k]$ is a coloring of a set of rectangles $\R$ with color set $[k]$. Then for a positive integer $w$ and color $i\in [k]$, let $\R_{w,i}$ be the set of $k$ smallest-height rectangles among the rectangles of $\{R\in\R \mid w(R)=w, \col(R)=i\}$. In case $|\{R\in\R \mid w(R)=w, \col(R)=i\}|<k$, we set $\R_{w,i}=\{R\in\R \mid w(R)=w, \col(R)=i\}$. % with width $w$.
    We define $\reduce_k(\R) = \bigcup_{w \in w(\R), i\in [k]} \R_{w,i}$.
\end{definition}
Notice that $\reduce_k(\R)$ contains at most $k^2|w(\R)|$ elements: for every possible width and every possible color, the at most $k$ rectangles of this specific width and of smallest height. Also, we have the following very simple observation.

\begin{lemma}\label{lem:replacement}
    Suppose $\R$ is a set of rectangles and $\col\colon \R\to [k]$ is a coloring function such that $k'\leq k$ rectangles from $\R$ with pairwise different colors can be packed in a zone $Z\subseteq \mathbb{R}^2$. Then one can also pack in $Z$ a set of $k'$ rectangles from $\reduce_k(\R,\col)$ with pairwise different colors.
\end{lemma}
\begin{proof}
    Let $\Q$ be the assumed packing of $k'\leq k$ rectangles from $\R$ of pairwise different colors in the zone $Z$. Note that if $\Q$ contains some rectangle of $R\in \R\setminus \reduce_k(\R,\col)$, then there exists another rectangle $R'\in \reduce_k(\R,\col)$ with $w(R')=w(R)$, $\col(R')=\col(R)$ and $h(R')\leq h(R)$ such that $R'$ was not used in the packing $\Q$. Hence, we can substitute $R$ with $R'$ in the packing $\Q$, fitting $R'$ within the area freed by removing $R$ from the packing. By applying such substitutions exhaustively, we obtain a packing in $Z$ consisting of $k'$ rectangles from $\reduce_k(\R,\col)$.
\end{proof}

Next, we use the following definitions to guess the polylines in a bottom to top order. For two monotone polylines $\P,\P'$ that start at the left side of $B$ and finish at the right side of $B$, we say that $\P'$ is {\em{below}} $\P$ (denoted by $\P' \leq \P$), if for every $x,y,y'$, $(x,y) \in \P$ and $(x,y')\in\P'$ implies $y'\leq y$.
We write $\P' < \P$ if $\P' \leq \P$ and $\P' \neq \P$. 
Given two polylines $\P'<\P$, we want to be able to solve the problem in the following sub-region:
\begin{definition}[$\container{\P',\P}$]
    For polylines $\P' < \P$, $\container{\P',\P}$ is the container (c.f. \Cref{def:container}) delimited by the box $B$, $\P'$ at the bottom and $\P$ at the top.
\end{definition}
Notice that if $\P$ has complexity $m$ and $\P'$ has complexity $m'$ then $\container{\P',\P}$ has complexity $m+m'+2$.

Now we give the algorithm in the case  when all rectangles in $\R$ have substantial width. This algorithm is encapsulated in the following lemma.

\begin{lemma}\label{lem:wide-instance}
    There is an algorithm that given $\eps>0$ and an instance $(B,\R,k)$ of {\sc{2D Knapsack}} in which all items are wide and have width at least $N_1/\alpha$, either returns a packing of size at least $(1-\eps)k$ or correctly concludes that there is no packing of size $k$. The running time is $(k+1/\eps)^{\O(k+1/\eps)}\cdot \alpha^{\O(k)}\cdot (|\R|\|B\|)^{\O(1/\eps^2)}$.
\end{lemma}
\begin{proof}
Let $\ell=N_1/(2\alpha)$; thus every rectangle on input has width at least $2\ell$. For clarity of presentation we allow the algorithm to output a packing of size at least $(1-3\eps)k$; then the result as stated in the lemma can be obtained by rescaling $\eps$ by factor $3$.

We first explain the algorithm. Compute $\F$ as given by \Cref{prop:perfect_hash}.
Then, guess (by trying all choices) a coloring  $\col\in \F$. The idea is now to use dynamic programming to compute a maximum-size structured packing for the colored instance.
More precisely, for every monotone polyline $\P$ of complexity at most $4/\eps+1$ connecting the left and the right side of $B$, and for every $C\subseteq [k]$, we shall compute the value $\DP{\P,C}$ defined as follows: $\DP{\P,C}$ is a maximum-size packing that contains only rectangles with colors in $C$, is colored injectively by $\col$, and is placed entirely below $\P$ with the added constraint that it is a subset of some $(\eps,\ell)$-structured packing.

To compute the value $\DP{\P,C}$ for given $\P$ and $C$, we iterate over all polylines $\P'$ of complexity at most $4/\eps+1$ that are below $\P$.
Let $B' = \container{\P',\P}$ be the container between $\P$ and $\P'$.
Iterate over all $C' \subseteq C$; this is the set of colors guessed to be used in~$B'$. 
Let $\R' = \reduce_k(\roundell{\R},\col)$ be the reduced set of rounded rectangles, where colors are naturally inherited from $\R$ during rounding.
Compute the following packings:
\begin{itemize}
    \item $\S_1$ is the largest packing in $B'$ consisting of at most $2/\eps^2$ rectangles with pairwise different colors from $C'$. This packing can be computed in time $|\R|^{\O(1/\eps^2)}\cdot (1/\eps)^{\O(1/\eps^2)}$ by first guessing the set of rectangles participating in it, and then checking whether the packing can be realized using the algorithm of Lemma~\ref{lem:algo_bounded_rects}. 
    \item $\S_2$ is the largest packing in $B'$ consisting of at most $k$ rectangles from $\R'$ with pairwise different colors from $C'$. Again, this packing can be computed in time $|\R'|^{\O(k)}\cdot (k+1/\eps)^{\O(k)}$ by first guessing the set of rectangles participating in it, and then checking whether the packing can be realized using the algorithm of Lemma~\ref{lem:algo_bounded_rects}. 
\end{itemize}
Iterate over $\S\in \{\S_1,\S_2\}$, and keep as $\DP{\P,C}$ the set $\DP{\P',C\setminus C'} \cup \S$ of maximum size over all the sets iterated on.
Finally, as the solution to the overall problem, return $\DP{\P,[k]}$ where $\P$ is the top side of $B$, provided this packing has size at least $(1-3\eps)k$. Otherwise, return that there is no packing of size $k$.

This concludes the description of the algorithm. We are left with (i) analyzing its running time and (ii) arguing that in case there is a packing of size at least $k$, the algorithm will output a packing of size at least $(1-3\eps)k$.

Let us start with assertion (ii). For this, suppose there exists a packing $\S$ of size $k$. Since all rectangles of $\S$ have width at least $2\ell$, by \Cref{lem:structural} there exists an $(\eps,\ell)$-structured packing $\S'$ of size at least $(1-3\eps)k$. Further, by the properties of $\F$, there exists $\col\in \F$ such that $\col$ is injective on $\S'$.
Now, let $\P_1,\P_2,\ldots,\P_m$ be the polylines witnessing the structuredness of $\S'$, and let $\emptyset=C_0\subseteq C_1\subseteq C_2\subseteq \ldots\subseteq C_m\subseteq C_{m+1}=[k]$ be such that $C_i$ is the sets of colors used by the rectangles of $\S'$ lying below~$\P_i$, where $\P_0$ and $\P_{m+1}$ are the bottom and the top side of $B$, respectively. A straightforward inductive argument using the structuredness of $\S'$ and \Cref{lem:replacement} shows now that for $i=0,1,\ldots,m+1$, the cell $\DP{\P_i,C_i}$ will contain a packing of size at least as large as the number of rectangles of $\S'$ lying below~$\P_i$. Hence, the algorithm will return a packing of size at least $|\S'|\geq (1-3\eps)k$, as promised.

We are left with analyzing the running time. The number of different colorings $\col\in \F$ is $|\F|\leq 2^{\O(k)}\cdot \log |\R|$.
Further, observe that the number of different polylines considered by the algorithm is bounded by $\|B\|^{\O(1/\eps)}$ and there are $2^k$ different subsets of colors.
Hence, the total number of cells $\DP{\P,C}$ considered by the algorithm is bounded by $2^k\cdot \|B\|^{\O(1/\eps)}$.
As argued, the time spent on computing a single value of $\DP{\P,C}$ is bounded by $2^k\cdot \|B\|^{\O(1/\eps)}$ (the number of choices for $\P'$ and $C')$ times 
\[|\R|^{\O(1/\eps^2)}\cdot (1/\eps)^{\O(1/\eps^2)}+|\R'|^{\O(k)}\cdot (k+1/\eps)^{\O(k)}.\]
Observe now that the rectangles of $\roundell{\R}$ have at most $\O(N_1/\ell')$ different widths, where $\ell'=\ell^2/N_1$. Since $\ell=N_1/2\alpha$, we conclude that the total number of different widths of the rectangles of $\roundell{\R}$ is bounded by
\[\O(N_1/\ell')=\O(N_1^2/\ell^2)\leq \O(\alpha^2).\]
Therefore,
\[|\R'|=|\reduce_k(\roundell{\R,\col})|\leq \O(\alpha^2 k^2).\]
Putting everything together, we infer that the running time of the algorithm is bounded by 
\begin{align*}
    & 2^{\O(k)}\cdot \log |\R|\cdot 2^{\O(k)}\cdot \|B\|^{\O(1/\eps)}\cdot \left(|\R|^{\O(1/\eps^2)}\cdot (1/\eps)^{\O(1/\eps^2)} + (\alpha^2 k^2)^{\O(k)}\cdot (k+1/\eps)^{\O(k)}\right)\\
    \leq &  (k+1/\eps)^{\O(k+1/\eps^2)}\cdot \alpha^{\O(k)}\cdot (|\R|\|B\|)^{\O(1/\eps^2)},
\end{align*}
as promised.
\end{proof}

\subsection{Full algorithm}

We now present the complete algorithm, which essentially boils down to making a reduction to the case when all rectangles on input have width at least $2\ell$, where $\ell=N_1/(\delta(B)k^2)$.
In the next lemma, we explain how to perform this reduction at the cost of removing $\eps k$ rectangles from the packing.
\begin{lemma}\label{lem:large_width}
    Let $\eps >0$. Suppose there is an algorithm $\mathscr A$ that, given a {\sc{2D Knapsack}} instance $(B=[0,N_1]\times [0,N_2],\R,p)$ in which all items are wide and have width at least $N_1/(\delta q^2)$ and the aspect ratio of $B$ is $\delta$, returns a packing of size at least $(1-\eps)p$ or attests that there is no packing of size $p$ in time $f(p,q,\eps,\delta,\|B\|,|\R|)$.
    Then there is an algorithm $\mathscr B$ that, given a {\sc{2D Knapsack}} instance $(B=[0,N_1]\times [0,N_2],\R,k)$ in which all items are wide and the aspect ratio of $B$ is $\delta$, returns a packing of size $(1-2\eps)k$ or attests that there is no packing of size $k$ in time $f(k,k,\eps,\delta,\|B\|,|\R|)+(1/\eps+|\R|)^{\O(1/\eps)}$.
\end{lemma}
\begin{proof}
    We present the algorithm $\mathscr B$.
    Without loss of generality, we can assume $k > 1/\eps$, as otherwise the number of rectangles in the sought packing is at most $1/\eps$ and we can solve the problem in time $(1/\eps+|\R|)^{\O(1/\eps)}$ by applying Lemma~\ref{lem:algo_bounded_rects} to every $k$-tuple of rectangles in $\R$.
    
    Let $\cal W$ be the set of rectangles of $\R$ that have width at most $N_1/(\delta k^2)$, and let $w=|\cal W|$. Note that since all rectangles are wide, the rectangles of $\cal W$ also have height bounded by $N_1/(\delta k^2)$.
    If $w\geq k$, then we can immediately construct a packing of size $k$ by stacking any $k$ rectangles of $\cal W$ vertically: they fit in the vertical dimension, because $k\cdot N_1/(\delta k^2)\leq N_2$.
    Otherwise, let $k'=k-w$.
    Run $\mathscr A$ on a modified instance where all rectangles of $\cal W$ are removed, with parameter $k'$. If there is no packing of size $k'$ for this instance, then clearly there is no packing of size $k$ for the original instance, and this conclusion may be reported by the algorithm.
    Otherwise, $\mathscr B$ returns a packing $\S'$ of size at least $(1-\eps)k'$ consisting of rectangles from $\R\setminus \cal W$.
    If $\S'$ consists only of rectangles of height at most $N_1/(\delta k)$, then we can again immediately obtain a  packing of size $k$ by stacking the rectangles of $\S'\cup \cal W$ vertically; again they fit in the vertical dimension, because $k\cdot N_1/(\delta k)\leq N_2$.
    Otherwise, we modify $\S'$ by removing any single rectangle $R$ present in $\S'$ whose height (and therefore also width) is at least $N_1/(\delta k)$, and putting all the rectangles of $\cal W$ into the space freed by the removal of $R$, by simply stacking them horizontally. They fit horizontally because $w\cdot N_1/(\delta k^2)\leq k\cdot N_1/(\delta k^2)=N_1/(\delta k)\leq w(R)$, and their heights are not greater than the height of $R$. The obtained modified packing $\S'$ is returned by the algorithm.

    It is clear that the algorithm outputs a packing and that when it concludes that there is no packing of size $k$, this conclusion is correct. What remains to show is that the packing eventually output by the algorithm has always size at least $(1-2\eps)k$. And indeed, the algorithm always is able to pack all rectangles packed in $\S'$, except for possibly one rectangle removed to accommodate $\cal W$, and all rectangles of $\cal W$. Hence, the packing output by the algorithm has always size at least
    \[(1-\eps)k'-1+w=(1-\eps)k-(1-\eps)w-1+w\geq (1-\eps)k-1> (1-2\eps)k,\]
    because $\eps k>1$ due to $k>1/\eps$.
\end{proof}

Now, \Cref{thm:main} follows immediately by combining the algorithm of \Cref{lem:wide-instance} with the reduction of \Cref{lem:large_width}. Observe that the running time is $\delta(B)^{\O(k)}\cdot (k+1/\eps)^{\O(k+1/\eps^2)}\cdot (|\R|\|B\|)^{\O(1/\eps^2)}$, as promised.

\section{Conclusion}

The correctness of our entire algorithm heavily relies on the assumption that every input rectangle is wide.
Indeed, this assumption is used in the greedy arguments in the proof of \Cref{lem:large_width}, which allows us to reduce to the case when every rectangle has a substantial width: at least $N_1/\mathrm{poly}(\delta(B),k)$. This assumption is again heavily used later on: in the proof of \Cref{lem:structural} it ensures that upon removing the rectangles corresponding to an $st$-separator in the conflict graph, there is enough space available for vertical shifting. This eventually leads to rounding the rectangles so that there are only $\mathrm{poly}(\delta(B),k)$ different possible widths, and thus effectively bounding the number of candidate rectangles to $\mathrm{poly}(\delta(B),k)$. So while the original problem --- the existence of a parameterized approximation scheme for {\sc{2D Knapsack}} --- remains open, we hope that the new structural techniques proposed in this work might give insight leading to its resolution.

\bibliography{bibliography}{}
\bibliographystyle{plain}

\end{document}